\documentclass[a4paper,USenglish]{llncs}

\usepackage{fullpage}
\usepackage{microtype}
\usepackage{thm-restate}
\usepackage{ntheorem}
\theoremstyle{plain}
\theoremsymbol{}
\theoremprework{}
\theorempostwork{}
\theoremseparator{.}

\renewcommand{\subparagraph}[1]{\paragraph{#1}}

\newtheorem{observation}{Observation}
\renewtheorem{theorem}{Theorem}
\renewtheorem{lemma}{Lemma}
\renewtheorem{claim}{Claim}

\usepackage{xcolor}
\usepackage{caption}
\usepackage{subcaption}

\usepackage{amsmath}
\usepackage{amssymb}
\newcommand{\qedhere}{\qed}

\usepackage[normalem]{ulem}
\usepackage{framed}

\def\lowhook#1{\lower.47em\hbox{$#1$}}
\def\LEorL{\mathrel{\lowhook\lhook\kern-.1em{\le}\kern-.08em\lowhook\rhook}}
\def\GEorG{\mathrel{\lowhook\lhook\kern-.08em{\ge}\kern-.1em\lowhook\rhook}}

\let\doendproof\endproof
\renewcommand\endproof{~\hfill\qed\doendproof}

\usepackage{dsfont}
\usepackage{xspace}
\usepackage{hyperref}
\usepackage{floatrow}
\usepackage{rotating}
\usepackage[noend]{algpseudocode}
\usepackage{stmaryrd}
\usepackage{stackengine}
\usepackage{scalerel}

\usepackage[ruled,vlined]{algorithm2e}
\SetKwInput{Input}{Input\rule{1.99ex}{0cm}}
\SetKwInput{Output}{Output}
\SetKwInput{Data}{Data}
\SetKwFor{Los}{}{}{}

\renewcommand{\paragraph}[1]{\smallskip\noindent\textbf{#1}\xspace}
\graphicspath{{./}}

\usepackage[nameinlink,capitalise]{cleveref}
\Crefname{observation}{Observation}{Observations}
\Crefname{algorithm}{Algorithm}{Algorithms}
\Crefname{section}{Section}{Sections}
\Crefname{observation}{Observation}{Observations}
\Crefname{lemma}{Lemma}{Lemmas}
\Crefname{claim}{Claim}{Claims}
\Crefname{figure}{Fig.}{Figs.}
\Crefname{figure}{Fig.}{Figs.}
\Crefname{enumi}{Condition}{Conditions}

\usepackage{lipsum}
\usepackage{dsfont}
\usepackage[color=yellow]{todonotes}
\usepackage{verbatim}

\usepackage{microtype}
\usepackage{xspace}
\usepackage{paralist}
\usepackage{dsfont}
\usepackage{amssymb}
\usepackage{gensymb}
\usepackage{wrapfig}

\newcommand{\RGB}{\ensuremath{\{r,g,b\}}}

\newcommand{\ytop}{\text{top}}
\newcommand{\leftcornersymb}{\raisebox{-1pt}{\includegraphics[page=2,scale=0.5]{symbols}}\hspace{1pt}}
\newcommand{\topcornersymb}{\raisebox{-1pt}{\includegraphics[page=3,scale=0.5]{symbols}}\hspace{1pt}}
\newcommand{\rightcornersymb}{\raisebox{-1pt}{\includegraphics[page=1,scale=0.5]{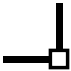}}\hspace{1pt}}

\newcommand{\leftcorner}[2]{\ensuremath{\leftcornersymb_{#1}(#2)}\xspace}
\newcommand{\rightcorner}[2]{\ensuremath{\rightcornersymb_{#1}(#2)}\xspace}
\newcommand{\topcorner}[2]{\ensuremath{\topcornersymb_{#1}(#2)}\xspace}
\newcommand{\hsidesymb}{\raisebox{0pt}{\includegraphics[page=5,scale=0.5]{symbols}}\hspace{1pt}}
\newcommand{\dsidesymb}{\raisebox{0pt}{\includegraphics[page=4,scale=0.5]{symbols}}\hspace{1pt}}
\newcommand{\vsidesymb}{\raisebox{0pt}{\includegraphics[page=6,scale=0.5]{symbols}}\hspace{1pt}}
\newcommand{\hside}[2]{\ensuremath{\hsidesymb_{#1}(#2)}\xspace}
\newcommand{\dside}[2]{\ensuremath{\dsidesymb_{#1}(#2)}\xspace}
\newcommand{\vside}[2]{\ensuremath{\vsidesymb_{#1}(#2)}\xspace}

\DeclareMathOperator{\DAG}{DAG}
\DeclareMathOperator{\RT}{RT}

\makeatletter
\renewcommand{\todo}[2][]{\@bsphack\@todo[#1]{\textcolor{black}{#2}}\@esphack\ignorespaces}
\makeatother

\newcommand{\remove}[1]{}

\usepackage{color}
\definecolor{realblue}{rgb}{0,0,1}
\definecolor{blue}{rgb}{0.274,0.392,0.666}
\definecolor{darkerblue}{rgb}{0.094,0.455,0.804}
\definecolor{darkblue}{rgb}{0.063,0.306,0.545}
\definecolor{red}{rgb}{0.627,0.117,0.156}
\definecolor{green}{rgb}{0,0.588,0.509}
\definecolor{orange}{rgb}{0.903,0.739,0.382}
\definecolor{realred}{rgb}{1,0,0}

\definecolor{lipicsblue}{rgb}{0.08235294118,0.3098039216,0.537254902}
\definecolor{ourred}{rgb}{1,0.3,0.3}
\definecolor{darkgreen}{rgb}{0, 0.7, 0}

\newcommand{\darkblue}[1]{{{\textcolor{darkblue}{#1}\xspace}}}

\renewcommand{\emph}[1]{\darkblue{\em #1}}

\hypersetup{colorlinks,linkcolor={lipicsblue},citecolor={lipicsblue},urlcolor={lipicsblue}} 

\begin{document}

\newcommand{\tuba}{$^1$}
\newcommand{\wurz}{$^2$}
\newcommand{\konst}{$^3$}
\newcommand{\rome}{$^4$}

\newcommand{\acks}{
This work was supported in part by DFG grant Ka812/17-1 and by MIUR-DAAD Joint Mobility Program n.57397196 (Angelini), 
in part by DFG grant WO$\,$758/11-1 (Chaplick),
and
in part by MIUR Project ``MODE'' under PRIN 20157EFM5C, 
by MIUR Project ``AHeAD'' under PRIN 20174LF3T8, 
by MIUR-DAAD JMP N$^\circ$ 34120, and
by H2020-MSCA-RISE project 734922 – ``CONNECT'' (Da Lozzo and Roselli).
\xspace}

\title{Morphing Contact Representations of Graphs\thanks{\acks}}

\author{
{Patrizio Angelini}\tuba,  
Steven Chaplick\wurz,
Sabine Cornelsen\konst,
{Giordano {Da Lozzo}}\rome, and
{Vincenzo {Roselli}}\rome
}
\institute{
University of T\"ubingen, T\"ubingen, Germany\\
\href{mailto:angelini@informatik.uni-tuebingen.de}{angelini@informatik.uni-tuebingen.de}\\
\smallskip
University of W\"urzburg, W\"urzburg, Germany\\
\href{mailto:steven.chaplick@uni-wuerzburg.de}{steven.chaplick@uni-wuerzburg.de}\\
\smallskip
University of Konstanz, Konstanz, Germany\\
\href{mailto:sabine.cornelsen@uni-konstanz.de}{sabine.cornelsen@uni-konstanz.de}\\
\smallskip
Roma Tre University, Rome, Italy\\
\href{mailto:giordano.dalozzo@uniroma3.it,vincenzo.roselli@uniroma3.it}{\{giordano.dalozzo,vincenzo.roselli\}@uniroma3.it}
}

\let\oldaddcontentsline\addcontentsline
\def\addcontentsline#1#2#3{}
\maketitle
\def\addcontentsline#1#2#3{\oldaddcontentsline{#1}{#2}{#3}}

\begin{abstract}
  We consider the problem of morphing between contact representations
  of a plane graph.  In an \emph{$\mathcal F$-contact representation}
  of a plane graph $G$, vertices are realized by internally disjoint
  elements from a family $\mathcal F$ of connected geometric objects.
  Two such elements touch if and only if their corresponding vertices
  are adjacent.  These touchings also induce the same embedding as in
  $G$.  In a morph between two $\mathcal F$-contact representations we
  insist that at each time step (continuously throughout the morph) we
  have an $\mathcal F$-contact representation.

  We focus on the case when $\mathcal{F}$ is the family of triangles
  in $\mathbb{R}^2$ that are the lower-right half of axis-parallel
  rectangles.  Such \emph{RT-representations} exist for every plane
  graph and right triangles are one of the simplest families of shapes
  supporting this property.  Thus, they provide a natural case to
  study regarding morphs of contact representations of plane graphs.

  We study \emph{piecewise linear morphs}, where each step is a
  \emph{linear morph} moving the endpoints of each triangle at
  constant speed along straight-line trajectories.  We provide a
  polynomial-time algorithm that decides whether there is a piecewise
  linear morph between two RT-representations of an $n$-vertex plane
  triangulation, and, if so, computes a morph with $\mathcal O(n^2)$
  linear morphs.  As a direct consequence, we obtain that for
  $4$-connected plane triangulations there is a morph between every
  pair of RT-representations where the ``top-most'' triangle in both
  representations corresponds to the same vertex.  This shows that the
  realization space of such RT-representations of any $4$-connected
  plane triangulation forms a connected set.
\end{abstract}

\section{Introduction}

We consider the morphing problem from the perspective of geometric representations of graphs. 
While a lot of work has been done to understand how to planarly morph the standard \emph{node-link} diagrams\footnote{where vertices are represented as points and edges as non-crossing curves} of plane graphs\footnote{the set of faces and the outer face are fixed} and to ``rigidly'' morph\footnote{scaling the objects is not allowed, e.g., as in \emph{bar-joint} systems~\cite{Laman1970} or in \emph{body-hinge} systems~\cite{DBLP:conf/gd/BowenDLR0T15,DBLP:journals/dcg/ConnellyDDFLMRR10,DBLP:books/daglib/0019278}} configurations of geometric objects, comparatively little has been explicitly done regarding (non-rigid) morphing of alternative representations of planar graphs, e.g., contact systems of geometric objects such as disks or triangles. 
In this case, the planarity constraint translates into the requirement of continuously maintaining a representation of the appropriate type throughout the morph. 

More formally, let $\mathcal F$ be a family of geometric objects homeomorphic to a disk. An \emph{$\mathcal F$-contact representation} of a plane graph $G$ maps vertices to internally disjoint elements of $\mathcal F$. We denote the geometric object representing a vertex $v$ by $\Delta(v)$. Objects $\Delta(v)$ and $\Delta(w)$ touch if and only if $\{v,w\}$ is an edge. 
The contact system of the objects must induce the same faces and outer face as in $G$.
A \emph{morph} 
between two $\mathcal F$-contact representations $R_0$ and $R_1$ of a plane graph $G$ is a continuously changing family of $\mathcal F$-contact representations $R_t$ of $G$ indexed by time $t \in [ 0,1 ]$.
An implication of the existence of morphs between any two representations of the same type is that the topological space defined by such representations is connected.
We are interested in elementary morphs, and in particular in \emph{linear morphs}, where the boundary points of the geometric objects move at constant speed along straight-line trajectories from their starting to their ending position. 
A \emph{\mbox{piecewise linear morph of length $\ell$}} between two
$\mathcal F$-contact representations $R_1$ and $R_{\ell+1}$ of a plane
graph $G$ is a sequence $\langle R_1,\dots,R_{\ell+1} \rangle$
of $\mathcal F$-contact representations of $G$ such that
$\langle R_i,R_{i+1} \rangle$ is a linear morph, for $i = 1,\dots,\ell$.
For a background on the mathematical aspects of morphing, see, e.g.,~\cite{alt2000}.

\subparagraph{Morphs of Node-Link Diagrams.}
F\'ary's theorem tells us that every plane graph has a node-link diagram where the edges are mapped to line segments. 
Of course, for a given plane graph $G$, there can be many such node-link diagrams of $G$, and the goal of the work in planar morphing is to study how (efficiently) one can create a smooth (animated) transition from one such node-link diagram to another while maintaining planarity.
Already in the 1940's Cairns~\cite{cairns1944} proved that, for plane triangulations, planar morphs exist between any pair of such node-link diagrams. 
However, the construction involved exponentially-many morphing steps. 
Floater and Gotsman~\cite{floater1999}, and  Gotsman and Surazhsky~\cite{DBLP:journals/cg/GotsmanS01,DBLP:journals/tog/SurazhskyG01} gave a different approach via Tutte's graph drawing algorithm~\cite{Tutte1963}, but this involves non-linear trajectories of unbounded complexity. 
Thomassen~\cite{thomassen1983} and Aranov et al.~\cite{aronov1993} independently showed that two node-link diagrams of the same plane graph have a \emph{compatible triangulation}\footnote{i.e., a way to triangulate both diagrams to produce the same plane triangulation} thereby lifting Cairns' result to plane graphs.
Of particular interest is the study of \emph{linear morphs}, where each vertex moves at a uniform speed along a straight-line.
After several intermediate results to improve the complexity of the morphs~\cite{DBLP:conf/soda/AlamdariACBFLPRSW13,DBLP:conf/icalp/AngeliniLBFPR14} and to
remove the necessity of computing compatible triangulations~\cite{DBLP:conf/gd/AngeliniFPR13}, the current state of the art~\cite{AlamdariABCLBFH17} is that there is a planar morph between any pair of node-link diagrams of any $n$-vertex plane graph using $\mathcal \theta(n)$ linear steps.
Planar morphs of other specialized plane node-link diagrams have also been considered, e.g., planar orthogonal drawings~\cite{Biedl:2013:MOP:2533288.2500118,vangoethem_et_al:LIPIcs:2018:8755}, convex drawings~\cite{DBLP:conf/compgeom/AngeliniLFLPR15}, upward planar drawings~\cite{ddfpr-upm-gd-18}, and so-called \emph{Schnyder drawings}~\cite{barreraCruz/haxel/lubiw:18}. In this latter result the lattice structure of all Schnyder woods of a plane triangulation~\cite{brehm:diplom,felsner:04} is exploited in order to obtain a sequence of linear morphs within a grid of quadratic size. Finally, planar morphs on the surface of a sphere~\cite{DBLP:journals/jgaa/KobourovL08} and in three dimensions have been investigated~\cite{abcdd-pd3-gd-18}. 

\subparagraph{Morphs of Contact Representations.}
Similar to F\'ary's theorem, the well-known Koebe-Andreev-Thurston theorem~\cite{Andreev,Koebe} states that every plane graph $G$ has a \emph{coin} representation, i.e.\ an $\mathcal F$-contact representation where $\mathcal F$ is the set of all disks. 
Additionally, for the case of 3-connected plane graphs, such coin representations of $G$ are unique up to \emph{M\"obius transformations}~\cite{BrightwellS93} -- see \cite{fr-pdcr-18} for a modern treatment. 
There has been a lot of work on how to intuitively understand and animate such transformations (see, e.g., the work of Arnold and Rogness~\cite{ArnoldRogness2008}), i.e., for our context, how to morph between two coin representations. 
Of course, ambiguity remains regarding how to formalize the complexity
of such morphs. 
In particular, this connection to M\"obius transformations appears to indicate that a theory of piecewise linear morphing for coin representations would be quite limited.

For this reason, we instead focus on contact representations of convex polygons.
These shapes still allow for representing all plane triangulations, as a direct consequence of the Koebe-Andreev-Thurston theorem, but are more amenable to piecewise linear morphs, where the linearity is defined on the trajectories of the corners.
De Fraysseix et al.~\cite{fraysseix/mendez/rosenstiehl:94} showed that every plane graph $G$ has a contact representation by triangles, and observed that these triangle-contact representations correspond to the 3-orientations (i.e., the \emph{Schnyder woods}) of $G$. Schrenzenmaier~\cite{schrezenmaier:wg17} used Schnyder woods to show that each 4-connected triangulation has a contact representation with homothetic triangles.
%
Gon{\c{c}}alves et al.~\cite{DBLP:journals/dcg/GoncalvesLP12} extended the triangle-contact results from triangulations~\cite{fraysseix/mendez/rosenstiehl:94} to 3-connected plane graphs, by showing that Felsner's generalized Schnyder woods~\cite{felsner:04} correspond to \emph{primal-dual} triangle-contact representations. 
Note that triangles and coins are not the only families of shapes that have been studied from the perspective of contact representations. 
Some further examples include boxes in $\mathbb{R}^3$~\cite{DBLP:conf/wg/ChaplickKU13,Felsner:2011,Thomassen86}, line segments~\cite{DBLP:journals/algorithmica/FraysseixM07,kuv-soda2013}, and homothetic polygons~\cite{DBLP:conf/isaac/LozzoDEJ17,DBLP:conf/wg/FelsnerSS18,schramm1990}. 

The construction of triangle-contact representations~\cite{fraysseix/mendez/rosenstiehl:94} (and the correspondence to 3-orientations) can be adjusted so that each triangle is the lower-right half of an axis-parallel rectangle. These \emph{right-triangle representations (RT-representations)} are our focus; see \cref{FIG:schnyder-wood}. 

\subparagraph{Our Contribution and Outline.} The paper is
organized as follows. We start with some definitions in
\cref{SEC:preliminaries} and describe the relationship between
(degenerate) RT-representations and Schnyder woods of plane
triangulations in \cref{SEC:RT}. 
In \cref{SEC:charRTmorph}, we provide necessary and sufficient conditions for a linear morph between two RT-representations. The first condition is that each corner $c$ of a triangle touches the same side $s$ of another triangle in the two representations, that is, the morph happens within the same Schnyder wood. The contact between $c$ and $s$ is always maintained when $s$ has the same slope in the two RT-representations. Otherwise, we require the point of $s$ hosting $c$ to be defined by the same convex combination of the end-points of $s$ in both representations.
In \cref{SEC:morphing_algorithm}, we present our morphing algorithm. If the two input RT-representations correspond to different Schnyder woods, we consider a path between them in the lattice structure of all Schnyder woods, similar to~\cite{barreraCruz/haxel/lubiw:18}, that satisfies some properties (if it exists). When moving along this path, from a Schnyder wood to another, we construct intermediate RT-representations that simultaneously correspond to both woods. We provide an algorithm to construct such intermediate RT-representations that result in a linear morph at each step.
Finally, in \cref{SEC:decision_algorithm}, we show how to decide whether there is a path in the lattice structure that satisfies the required properties.
This results in an efficient testing algorithm for the existence of a piecewise linear morph between two RT-representations of a plane triangulation; in the positive case, the computed piecewise linear morph has at most quadratic length.
Consequently, for $4$-connected plane triangulations, under a natural condition on the outer face of their RT-representations, the topological space defined by such RT-representations is connected.

\smallskip
Full details for omitted or sketched proofs can be found in the Appendix.

\section{Definitions and Preliminaries}\label{SEC:preliminaries}

\paragraph{Basics.}
A \emph{plane triangulation} is a maximal planar graph with a distinguished outer face. 
A \emph{directed acyclic graph (DAG)} is an oriented graph with no directed cycles. A
\emph{topological ordering} of an $n$-vertex DAG $G=(V,E)$ is a
one-to-one map $\tau: V \rightarrow \mathds \{1,\dots,n\}$ such
that $\tau(v) < \tau(w)$ for $(v,w) \in E$.
Let $p$ and $q$ be two points in the plane. The \emph{line segment}
$\overline{pq}$ is the set
$\{(1-\lambda) p + \lambda q;\;0 \leq \lambda \leq 1\}$ of convex
combinations of $p$ and $q$. Considering $\overline{pq}$
\emph{oriented} from $p$ to $q$, we say that \emph{$x$ cuts $\overline{pq}$ with the ratio $\lambda$} if $x= (1-\lambda) p + \lambda q$.

\smallskip
In the case of polygons, a linear morph is
completely specified by the initial and final positions of the corners
of each polygon. If a corner $p$ is at position $p_0$ in the initial
representation (at time $t = 0$) and at position $p_1$ in the final
representation (at time $t = 1$), then its position at time $t$ during
a linear morph is $(1 - t)p_0 + t p_1$ for any $0 \leq t \leq 1$.

\smallskip
\paragraph{Schnyder Woods.}
\begin{figure}[t]
\centering
\begin{subfigure}[b]{.3\textwidth}\centering
\includegraphics[page=5]{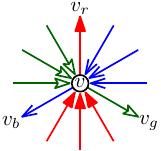}
\caption{\label{fig:schnyderInner}inner vertex}
\end{subfigure}\hfil
\begin{subfigure}[b]{.3\textwidth}\centering
\includegraphics[page=6]{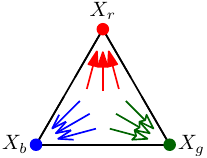}
\caption{\label{fig:schnyderOuter}outer vertices}
\end{subfigure}
\caption{\label{FIG:schnyderWood}The two conditions for a Schnyder wood.}
\end{figure}
A \emph{3-orientation} \cite{brehm:diplom,felsner:04} of a plane
triangulation is an orientation of the inner edges such that each
inner vertex has out-degree 3 and the three outer vertices have
out-degree $0$. A \emph{Schnyder wood} $T$ \cite{schnyder:soda90} of a
plane triangulation $G$ is a 3-orientation together with a partition
of the inner edges into three color classes, such that the three
outgoing edges of an inner vertex have distinct colors and all the incoming edges of an outer vertex have the same
color. Moreover, the color assignment around the vertices must be as
indicated in~\cref{FIG:schnyderWood}. We say that a cycle in a
Schnyder wood is \emph{oriented} if it is a directed~cycle.

The following well-known properties of Schnyder woods can directly be
deduced from the work of Schnyder~\cite{schnyder:soda90}.
  \begin{inparaenum}
\item\label{ITEM:SchnyderWoodsExist} Every plane triangulation has a 3-orientation. 
\item\label{ITEM:3oriEqualsSchnyderWood} For each 3-orientation of a plane triangulation there is exactly one partition of the inner edges into three color classes such that the pair yields a Schnyder wood. 
\item\label{ITEM:IsWood}
  Each color class of a Schnyder wood
  induces a directed spanning tree rooted at an outer vertex. 
   \item\label{ITEM:SchnyderDAG}
  Reversing the edges of two color classes
  and maintaining the orientation of the third color class yields a
  directed acyclic graph.
   \item\label{ITEM:orientedTriangle3colored}
  The edges of an oriented
  triangle in a Schnyder wood have three distinct colors and every
  triangle composed of edges of three different colors is oriented.
\end{inparaenum}

We call the color classes red
(r,\includegraphics[page=55]{figures.pdf}), blue
(b,\includegraphics[page=56]{figures.pdf}), and green
(g,\includegraphics[page=57]{figures.pdf}).  
The symbols $X_r$, $X_b$, and $X_g$
 denote the \emph{red}, \emph{blue}, and \emph{green
  outer vertex} of $G$, i.e., the outer vertices with
incoming red, blue, and green edges, respectively.
For an inner vertex $v$, let $v_r$, $v_b$, and $v_g$ be the respective neighbors of $v$ such that $(v,v_r)$ is red, $(v,v_b)$ is blue, and $(v,v_g)$ is
green. Finally, let $\DAG_r(T)$ ($\DAG_b(T)$) be the directed acyclic
graph obtained from $G$ by orienting all red (blue) edges as in $T$ while all blue (red) and green edges are reversed.

Let $C$ be an oriented triangle of a Schnyder wood
$T$. Reversing $C$ yields another 3-orientation with its unique
Schnyder wood $T_C$. If $C$ is a facial cycle, then $T$ differs from
$T_C$ by recoloring the edges on $C$ only.  More precisely, the former
outgoing edge of a vertex gets the color of the former incoming edge
of the same vertex. This procedure of reversing and recoloring is
called \emph{flipping}\footnote{Brehm~\cite{brehm:diplom} called
  flipping a counter clockwise triangle a flip, and flipping a
  clockwise triangle a flop.} an oriented triangle of a Schnyder wood.
Any Schnyder wood can be converted into any other Schnyder wood
of the same plane triangulation by flipping $\mathcal O(n^2)$ oriented
triangles~\cite{barreraCruz/haxel/lubiw:18,brehm:diplom}.
For two Schnyder woods $T_0$ and $T_\ell$, 
$\mathcal<C_1,\dots,C_{\ell}\mathcal>$ is a \emph{flip sequence
between $T_0$ and $T_\ell$} if there are Schnyder
woods $T_1,\dots,T_{\ell-1}$ such that $C_i$, $i=1,\dots,\ell$, is an
oriented triangle in $T_{i-1}$ and $T_i$ is obtained from $T_{i-1}$ by
flipping $C_i$.  
We say that \emph{a Schnyder wood $T'$ can be
  obtained from a Schnyder wood $T$ by a sequence of facial flips} if
there is a flip sequence between $T$ and $T'$ that contains only
facial cycles.

\section{RT-Representations of Plane Triangulations}\label{SEC:RT}

Let $R$ be an RT-representation of a plane triangulation $G$ and let $u$ be a vertex of $G$. Recall that $\Delta(u)$ is the triangle representing $u$ in $R$. We denote by \hside{}{u}, \vside{}{u}, and \dside{}{u} the \emph{horizontal}, \emph{vertical}, and \emph{diagonal} side of $\Delta(u)$.  
Further, we denote by \leftcorner{}{u}, \rightcorner{}{u}, and \topcorner{}{u}, the left, right, and top corner of $\Delta(u)$, respectively.
If two triangles touch each other in their corners, we say that these two corners \emph{coincide}. If there exist no two triangles whose corners coincide, then $R$ is \emph{non-degenerate}; otherwise, it is \emph{degenerate}.
Let $(c,s)$ be a pair with $c\in \{\leftcornersymb, \rightcornersymb, \topcornersymb \}$ and $s \in \{\vsidesymb,\hsidesymb,\dsidesymb\}$, we say that $(c,s)$ is a \emph{compatible pair} if it belongs to the set $\{ (\rightcornersymb, \dsidesymb), (\leftcornersymb, \vsidesymb), (\topcornersymb, \hsidesymb)\}$.
Observe that, in any RT-representation of $G$, if a corner $c$ of a triangle $\Delta(u)$ touches the side $s$ of a triangle $\Delta(v)$, with $(u,v) \in E(G)$, then $(c,s)$ is a compatible pair. 
We formally require this also in the case of a degeneracy. E.g., if $\rightcorner{}{v}$ coincides with $\leftcorner{}{u}$ for two vertices $u$ and $v$, then the respective compatible pair is either $(\leftcornersymb, \vsidesymb)$ or $(\rightcornersymb, \dsidesymb)$~--~even though $\rightcorner{}{v}$ also touches $\hside{}{u}$, and  $\leftcorner{}{u}$ touches $\hside{}{v}$.

In the next two subsections, we describe the relationship between RT-representations and Schnyder woods~\cite{fraysseix/mendez/rosenstiehl:94} and extend it to the case of degenerate RT-representations.

\subsection{From RT-Representations to Schnyder Woods}
\label{se:RT2wood}

Let $G=(V,E)$ be a plane triangulation with a given
RT-representation $R$. It is possible to orient and color the edges of $G$ in
order to obtain a Schnyder wood by considering the types of contacts between triangles in $R$ as follows.

First, consider the non-degenerate case; refer to \cref{fig:yieldsWood}.
Let $e=\{u,v\} \in E$ be an inner edge such that a corner
$c$ of $\Delta(u)$ touches a side $s$ of $\Delta(v)$. 
We use the following rules: We orient $e$ from $u$ to $v$, and
color $e$
\begin{inparaenum}
\item[\textbf{blue}] if $c$ is \leftcorner{}{u},
\item[\textbf{green}] if $c$ is \rightcorner{}{u},
\item[\textbf{red}] if $c$ is  \topcorner{}{u}.
\end{inparaenum}

\begin{lemma}[\cite{fraysseix/mendez/rosenstiehl:94}, Theorem 2.2]
  The above assignment yields a Schnyder wood.
\end{lemma}

Assume now that there exist two triangles $\Delta(u)$ and $\Delta(v)$ whose corners coincide. Observe that the assignment of colors and directions to the edge $\{u,v\}$ determined by the procedure above would be ambiguous. 
The next observation will be useful to resolve this ambiguity.
\begin{restatable}{observation}{threeInACorner}\label{obs:threeInAcorner}\label{obs:threeInAcorner}
  In an RT-representation of a plane triangulation, if the corner of a triangle $\Delta(u)$ coincides with the corner of a triangle $\Delta(v)$ in a point $p$, then there exists a triangle $\Delta(w)$, $w\neq u,v$, with a corner on $p$, unless $\{u,v\}$ is an edge of the outer face.
\end{restatable}

By \cref{obs:threeInAcorner}, in a degenerate RT-representation there exist three vertices $u$, $v$, and $w$ such that $\topcorner{}{u}$, $\leftcorner{}{v}$, $\rightcorner{}{w}$ lie on a point, see \cref{FIG:commoncornercoloring}. 
For each of the three edges, a choice of coloring and orientation corresponds to deciding which of the two triangles participates to the touching with its corner and which triangle with an extremal point of one of its sides. 
This yields two options as indicated in \cref{FIG:commoncornercoloring}, both resulting in a Schnyder wood. Note that the face $f=\left<u,v,w\right>$ is cyclic in both these Schnyder woods, and each of them can be obtained from the other by flipping $f$.
\begin{figure}[t]
\centering
\begin{subfigure}[b]{.3\textwidth}\centering
\includegraphics[page=2, scale=.7]{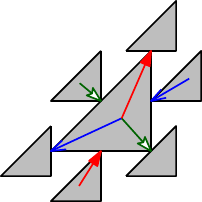}
    \subcaption{non-degenerate case}\label{fig:yieldsWood}
\end{subfigure}\hfil
\begin{subfigure}[b]{.6\textwidth}\centering
    \includegraphics[page=12]{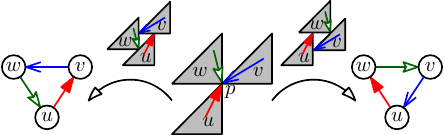}\linebreak\vfill
    \subcaption{degenerate case}\label{FIG:commoncornercoloring}
\end{subfigure}    
\caption{From an RT-representation to a Schnyder wood.}
\end{figure}
Summarizing, we get the following.
\begin{observation}\label{obs:lattice-dimension}
  Given an RT-representation $R$ of a plane triangulation $G$, let $P$
  be the set of points where three triangles meet. Then, $R$
  corresponds to a set $\mathcal{T}_R$ of $2^{|P|}$ different Schnyder woods on $G$, the points
  of $P$ correspond to $|P|$ edge-disjoint oriented triangles, and the
  Schnyder woods in $\mathcal T_R$ differ in flipping some of them.
\end{observation}

\subsection{From Schnyder Woods to RT-Representations}
\label{SEC:wood2RT}

Assume now that we are given a Schnyder wood $T$ of a plane triangulation $G=(V,E)$. We describe a technique for constructing an RT-representation of $G$ corresponding to $T$ in which the y-coordinate of the horizontal side of each triangle is prescribed by a function $\tau : V \rightarrow \mathbb{R}$ satisfying some constraints; observe that in the non-degenerate case
in~\cite{fraysseix/mendez/rosenstiehl:94} $\tau$ is a topological labeling of $\DAG_r(T)$, i.e., a canonical ordering of $G$.

We call $\tau: V \longrightarrow \mathds R$
an \emph{Admissible Degenerate Topological} labeling of the graph $\DAG_r(T)$, for short \emph{ADT-labeling}, if for each directed edge $(u,v)$ of $\DAG_r(T)$, we have
\begin{inparaenum}
  \item \label{cond:tau-1} $\tau(u) \leq \tau(v)$ and
  \item \label{cond:tau-2} $\tau(u)=\tau(v)$ only if
  \begin{inparaenum}
    \item $(u,v)$ is green and belongs to a clockwise oriented facial cycle, or
    \item $(u,v)$ is blue and belongs to a counter-clockwise oriented facial cycle, and
          \end{inparaenum}
  \item \label{cond:tau-3} if $\tau(u_b) = \tau(u) = \tau(u_g)$ for a vertex $u$, and $u_1$ and $u_2$ are vertices such that $\left<u,u_g,u_1\right>$ is a clockwise facial cycle and  $\left<u,u_b,u_2\right>$ is a counter-clockwise facial cycle, then $u_1 \neq u_2$.
      \end{inparaenum}
      
\begin{lemma}\label{lem:from-rt-to-tau}
Let $R$ be an RT-representation of a plane triangulation $G=(V,E)$, let $T$ be a Schnyder wood corresponding to $R$, and let $\tau(v)$, $v \in V$, be the y-coordinate of $\hside{}{v}$. Then, $\tau$ is an ADT-labeling of $\DAG_r(T)$.
\end{lemma}

\begin{proof}
Let $(u,v)$ be a directed edge of $\DAG_r(T)$. By the definition of $T$, we get immediately that $\tau(u) \leq \tau(v)$ independently of whether $(u,v)$ is red, green, or blue. In fact, if $(u,v)$ is red, then it is oriented from $u$ to $v$ in $T_r$. Thus, the compatible pair corresponding to such an edge in $R$ is $(\topcorner{}{u},\hside{}{v})$. Hence, $\hside{}{u}$ lies strictly below $\hside{}{v}$.
If $(u,v)$ is green (resp., blue), then it is oriented from $v$ to $u$ in $T$. Thus, the compatible pair corresponding to such an edge in $R$ is $(\rightcorner{}{v},\dside{}{u})$ (resp., $(\leftcorner{}{v},\vside{}{u})$).
Hence, $\hside{}{u}$ does not lie above $\hside{}{v}$.

Assume that $\tau(u) = \tau(v)$, which implies that $(u,v)$ is not red, as observed above. Suppose that $(v,u)$ is a green edge. Then, $\rightcorner{}{v}$ and $\leftcorner{}{u}$ coincide.
By \cref{obs:threeInAcorner}, there exists a vertex $z$ such that $\topcorner{}{z}$ coincides with $\rightcorner{}{v}$ and $\leftcorner{}{u}$. Thus, $\left<v,u,z\right>$ is a clockwise oriented facial cycle. Similarly, when $(v,u)$ is a blue edge, there is a vertex $z$ such that $\leftcorner{}{v}$, $\rightcorner{}{u}$, and $\topcorner{}{z}$ coincide. Therefore, $\left<v,u,z\right>$ is a counter-clockwise oriented facial cycle.

Finally, if $\tau(u_b)=\tau(u)=\tau(u_g)$ and $u_1$ and $u_2$ are the vertices such that $\left<u,u_g,u_1\right>$ is a clockwise facial cycle and  $\left<u,u_b,u_2\right>$ is a counter-clockwise facial cycle, then $\topcorner{}{u_1}$ touches $\rightcorner{}{u}$ and $\topcorner{}{u_2}$ touches $\leftcorner{}{u}$. Thus, $u_1 \neq u_2$.
\end{proof}

\begin{restatable}{lemma}{fromTandTautoR}\label{lem:from-T-and-Tau-to-R}
Let $T$ be a Schnyder wood of an $n$-vertex plane triangulation $G$, let $\tau$ be an ADT-labeling of $\DAG_r(T)$, and let $R_o = \Delta(X_r) \cup \Delta(X_g) \cup \Delta(X_b)$ be an RT-representation of the outer face of $G$ such that $\hside{}{X_i}$ has y-coordinate $\tau(X_i)$, with $i\in \{r, g, b\}$. Then, there exists a unique RT-representation $\RT(T, \tau, R_o)$ of $G$ corresponding to $T$ in which $\hside{}{v}$ has y-coordinate $\tau(v)$, for each vertex $v$ of $G$, and in which the outer face is drawn as in $R_o$.
\end{restatable}

\paragraph{Outline of the proof.}
\begin{figure}[t]
\begin{subfigure}[b]{.28\textwidth}\centering
\includegraphics[page=10,scale=.76]{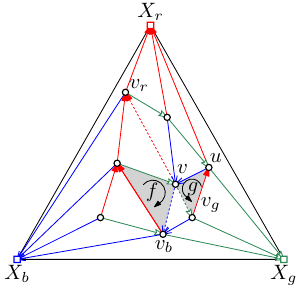}
\subcaption{Schnyder wood of $G$\label{FIG:schnyder-wood-trees}}
\end{subfigure}
\begin{subfigure}[b]{.28\textwidth}\centering
\includegraphics[page=13,scale=.76]{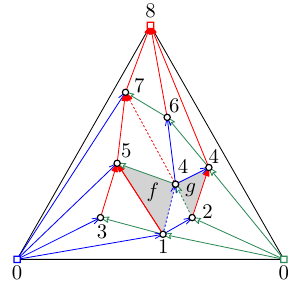}
\subcaption{Labeling $\tau$ of $G$\label{FIG:schnyder-wood-dag}}
\end{subfigure}
\begin{subfigure}[b]{.42\textwidth}\centering
\includegraphics[page=14,scale=.76]{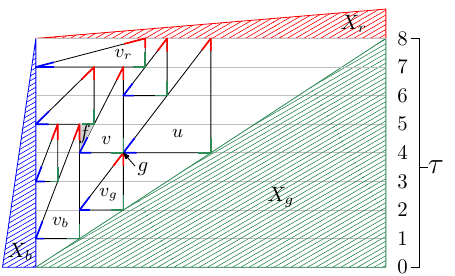}
\subcaption{$\RT(T, \tau, R_o)$\label{FIG:schnyder-wood-to-representation}}
\end{subfigure}
\caption{(a) A Schnyder wood $T$ of a plane triangulation $G$; the edges connecting $v$ with vertices $v_r$, $v_g$, and $v_b$ are dashed. (b) Graph $\DAG_r(T)$ with ADT-labeling $\tau$. (c) An RT-representation of $G$ constructed from $T$, $\tau$, and an RT-representation $R_o$ of the outer face.}
\label{FIG:schnyder-wood}
\end{figure}
We process the vertices of $G$ according to a topological ordering $\tau'$ of $\DAG_r(T)$. In the first two steps, we draw triangles $\Delta(X_b)$ and $\Delta(X_g)$ as in $R_o$; see  \cref{FIG:schnyder-wood-to-representation}. At each of the following steps, we consider a vertex $v$, with $2 < \tau'(v)=i < n$. 

We draw $\Delta(v)$ with its horizontal side on $y=\tau(v)$ and with its top corner at $y=\tau(v_r)$, as follows. Since the blue edge $(v_b,v)$ and the green edge $(v_g,v)$ are entering $v$ in $\DAG_r(T)$, the triangles $\Delta(v_b)$ and $\Delta(v_g)$ have already been drawn. Also, by \cref{cond:tau-1} of ADT-labeling, we have that $\tau(v_b)\leq \tau(v)$ and $\tau(v_g) \leq \tau(v)$. Further, it can be shown that $\topcorner{}{v_b}$ and $\topcorner{}{v_g}$ have y-coordinate larger than or equal to $\tau(v)$, and that if a triangle $\Delta(u)$ intersects the line $y=\tau(v)$ between $\Delta(v_b)$ and $\Delta(v_g)$, then $u$ is a neighbor of $v$ such that $v = u_r$. By construction, $\topcorner{}{u}$ has y-coordinate equal to $\tau(v)$. Thus, we draw the horizontal side of $\Delta(v)$ on $y=\tau(v)$ between $\Delta(v_b)$ and $\Delta(v_g)$. The conditions of ADT-labelings guarantee that $\hside{}{v}$ has positive length. If $i=n$, and hence $v = X_r$, we draw $\Delta(X_r)$ as in $R_o$. 
\qedhere

\section{Geometric Tools}\label{SEC:charRTmorph}

In this section, we provide geometric lemmata that will be exploited in the subsequent sections. 
We first show that the incidence of a point and a line segment is
maintained during a linear morph if the line segment is
moved in parallel (with a possible stretch, but keeping the
orientation) or the ratio with which the point cuts the segment is
maintained; see~\cref{FIG:morphing_segment_plus_point}.

\begin{lemma}
   \label{LEMMA:sidewithpoint}
   For $i=0,1$ let $p_i$, $q_i$ be two points in the plane and let
   $x_i \in \overline{p_iq_i}$.  For $0 < t < 1$, further let
   $p_t = (1-t) p_0 + t p_1$ and $q_t = (1-t) q_0 + t q_1$. Then,
   $x_t = (1-t) x_0 + t x_1 \in \overline{p_t q_t}$ if
   \begin{enumerate}
   \item $\overline{p_0 q_0}$ and $\overline{p_1 q_1}$ are parallel with 
the same direction, or
   \item $x_0$ cuts $\overline{p_0q_0}$ with the same ratio as $x_1$ 
cuts $\overline{p_1q_1}$
   \end{enumerate}
\end{lemma}
\begin{proof}
   Assume first that $\overline{p_0q_0}$ and $\overline{p_1q_1}$ are
   parallel.  If $\overline{p_0q_0}$ and $\overline{p_1q_1}$ are
   collinear, we may assume that they are both contained in the x-axis,
   that $p_i,q_i$, $i=0,1$, are real numbers, and that $p_0 <
   q_0$. Since $\overline{p_0q_0}$ and $\overline{p_1q_1}$ have the
   same direction, this implies that $p_1 < q_1$. Since $x_i$, $i=0,1$,
   is a point in $\overline{p_iq_i}$, it follows that
   $p_i \leq x_i \leq q_i$. Hence, we get for $t\in [0,1]$ that
   \[
     \underbrace{(1-t)p_0+t p_1}_{p_t} \leq \underbrace{(1-t)x_0+t
       x_1}_{x_t} \leq \underbrace{(1-t)q_0+t q_1}_{q_t}.
   \]
   If $\overline{p_0 q_0}$ and $\overline{p_1 q_1}$ are parallel with
   the same direction but not collinear, then the polygon
   $\left<p_0,q_0,q_1,p_1\right>$ is convex. Thus, $\overline{x_0 x_1}$
   must intersect $\overline{p_t q_t}$, for any $t$.  Also,
   $\overline{p_t q_t}$ and $x_t$ both lie on the same line
   $\ell_t$. More precisely, let $d$ be the distance between the lines
   through segments $\overline{p_0 q_0}$ and $\overline{p_1 q_1}$.
   Then, $\ell_t$ is the line with distance $td$ from
   $\overline{p_0 q_0}$.

   Finally, if $x_0$ cuts $\overline{p_0q_0}$ with the same ratio
   $\lambda$ as $x_1$ cuts $\overline{p_1q_1}$, then $x_t = (1-t)((1-\lambda) p_0 + \lambda q_0) + t ((1-\lambda) p_1 + \lambda 
 q_1) = (1-\lambda)p_t + \lambda q_t \in \overline{p_t q_t}$.
\end{proof}

\begin{figure}[t]
\centering
\begin{subfigure}[b]{.4\textwidth}\centering
\includegraphics[page=20]{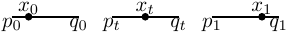}
\subcaption{\label{FIG:collinear}segments collinear}
\end{subfigure}%
\begin{subfigure}[b]{.3\textwidth}\centering
\includegraphics[page=18]{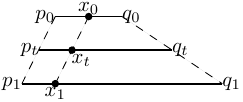}
\subcaption{\label{FIG:parallel1}segments parallel}
\end{subfigure}%
\begin{subfigure}[b]{.3\textwidth}\centering
\includegraphics[page=21]{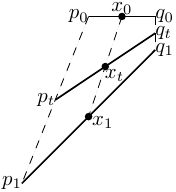}
\subcaption{\label{FIG:ratio}cut ratio maintained}
\end{subfigure}
\caption{\label{FIG:morphing_segment_plus_point}Morphing a segment and a point.}
\end{figure}

\cref{LEMMA:sidewithpoint} implies the following sufficient criterion for a linear morph.

\begin{restatable}{lemma}{charRTmorph}\label{THEO:charRTmorph}
   Let $R_0$ and $R_1$ be two RT-representations of a triangulation $G$ corresponding to the same Schnyder wood such that the triangles of the outer face pairwise touch in their corners.
   The pair $\left<R_0,R_1\right>$ defines a linear morph if, for any two adjacent vertices $u$ and $v$ such that a corner $c_i(v)$ of $v$ touches a side $s_i(u)$ of $u$, where $c\in \{\topcornersymb, \leftcornersymb, \rightcornersymb \}$ and $s \in \{ \hsidesymb, \vsidesymb, \dsidesymb \}$, one of the following holds:\\
     \begin{inparaenum}
     \item \label{COND:inCHAR_parallel}
       $s_1(u)$ and $s_2(u)$ are parallel. 
     \item \label{COND:inCHAR_same_ratio}
       $c_1(v)$ cuts $s_1(u)$ with the same ratio as $c_2(v)$~cuts~$s_2(u)$. 
     \end{inparaenum} 
\end{restatable}

By \cref{obs:lattice-dimension}, an RT-representation $R$ of a plane triangulation $G$ corresponds to a set $\mathcal T_R$ of Schnyder woods that differ from each other by flipping a set of edge disjoint triangles. The \emph{topmost} vertex of $R$ is the vertex $v$ of $G$ maximizing the y-coordinate of $\hside{}{v}$.

\begin{restatable}{lemma}{ledifferentflips}\label{le:triangulations-linear-morph}
  Let $R_0$ and $R_1$ be two RT-representations of the same plane
  triangulation $G=(V,E)$ such that $\left<R_0,R_1\right>$ is a linear morph. 
  Then $\mathcal T_{R_0} \cap \mathcal T_{R_1} \neq \emptyset$.
\end{restatable}

\begin{theorem}[Necessary Condition]\label{THEO:nec}
  If there is a piecewise linear morph between two RT-representations
  of a plane triangulation $G$, then the corresponding Schnyder woods
  can be obtained from each other by a sequence of facial flips.  In
  particular the topmost vertex is the same in both representations if
  $G$ has more than three vertices.
\end{theorem}

\begin{proof}
  Let $\left<R_1,\dots,R_\ell\right>$ be a sequence of linear
  morphs. \cref{le:triangulations-linear-morph} implies
  $\mathcal T_{R_i} \cap \mathcal T_{R_{i+1}} \neq \emptyset$ for
  $i=1,\dots,\ell-1$. Let
  $T_i \in \mathcal T_{R_i} \cap \mathcal T_{R_{i+1}}$ for
  $i=1,\dots,\ell-1$. Then $T_{i+1}$, $i=1,\dots,\ell-2$ can be
  obtained from $T_{i}$ by a sequence of edge-disjoint facial
  flips. Hence, the Schnyder wood $T_{\ell-1}$ of $R_\ell$ can be
  obtained from the Schnyder wood $T_1$ of $R_1$ by a sequence of
  facial~flips. 
\end{proof}

\section{A Morphing Algorithm}\label{SEC:morphing_algorithm}

In this section, we prove the following theorem.

\begin{theorem}[Sufficient Condition]\label{THEO:morph_exists}
  Let $R_1$ and $R_2$ be two RT-representations of an $n$-vertex plane
  triangulation $G$ corresponding to the Schnyder woods $T_1$ and $T_2$, respectively. If
  $T_2$ can be obtained from $T_1$ by a sequence of $\ell$ facial
  flips, then there exists a piecewise linear morph between $R_1$ and $R_2$ of length
  $\mathcal O(n + \ell)$. Such a morph can be computed in $\mathcal O(n(n+\ell))$ time, provided that the respective sequence of $\ell$ facial flips is given.
\end{theorem}

Since there is always a piecewise linear morph between two RT-representations of a plane triangle (see \cref{FIG:triangle}), we will assume that $G$ has at least four vertices. This implies especially that the topmost vertex, which always coincides with $X_r$, is the same in $R_1$ and $R_2$.
\begin{figure}[t]
  \centering
  \includegraphics[scale=.9]{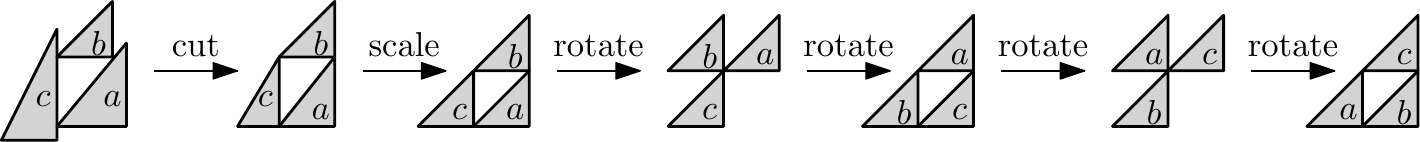}
  \caption{\label{FIG:triangle}Morphing an RT-representation of a triangle to a labeled canonical form: First, cut the extruding parts of the triangles, maintaining the slopes of the diagonal sides. Then, scale the triangles such that the horizontal and vertical sides have length one. Finally, keep rotating the triangles until the topmost vertex is as desired.}
\end{figure}

In \cref{SEC:moving_a_triangle}, we introduce our main procedure \textsc{adjust}, which moves a triangle in an RT-representation along an incident diagonal and adjusts the remaining triangles so that the result is a linear morph. Repeatedly applying \textsc{adjust}, we first morph $R_1$ to a non-degenerate RT-representation that still corresponds to $T_1$ (\cref{SEC:same-wood-different-proportions}); then, we perform a sequence of linear morphs to realize the $\ell$ facial flips geometrically (\cref{SEC:flipping-algorithm}), hence obtaining an RT-representation corresponding to $T_2$, which we finally morph to $R_2$ (\cref{SEC:same-wood-different-proportions}).

\subsection{Moving a Triangle Along a Diagonal}\label{SEC:moving_a_triangle}
Let $G=(V, E)$ be a plane triangulation and let $R$ be an RT-representation of $G$ corresponding to a Schnyder wood $T$ of $G$. Given an inner vertex $x$ of $G$ and a real value $y$ with some properties, \textsc{adjust} computes a new RT-representation $R'$ of $G$ corresponding to $T$ in which $\hside{}{x}$ has y-coordinate $y$ and $\Delta(x_g)$ remains unchanged, such that $\left<R,R'\right>$ is a linear morph. 

To achieve this goal, the y-coordinate of $\hside{}{v}$, for some vertex $v \neq x$, may also change; however, the ratio with which $\rightcorner{}{v}$ cuts $\dside{}{v_g}$ does not change, thus satisfying \cref{COND:inCHAR_same_ratio} of \cref{THEO:charRTmorph}. The y-coordinates of the horizontal sides are encoded by a new ADT-labeling $\tau$ of $G$, and $R'$ is the unique RT-representation $\RT(T, \tau, R_o)$ of $G$ that is obtained by applying \cref{lem:from-T-and-Tau-to-R} with input $G$, $T$, $\tau$, and the representation $R_o$ of the outer face of $G$ in $R$.

For a vertex $w \in V$, we denote by $\ytop(w)$ the y-coordinate of $\topcorner{}{w}$; recall that, in our construction, we have $\ytop(w) = \tau(w_r)$, if $w$ is an inner vertex. Also, let $v_1,\dots,v_\ell$ be the neighbors of $w$
such that $\rightcorner{}{v_1},\dots,\rightcorner{}{v_\ell}$ appear in this order from $\leftcorner{}{w}$ to $\topcorner{}{w}$ along $\dside{}{w}$. For a fixed $i \in \{1,\dots,\ell\}$, we say that \emph{moving $v_i$ to $y \in \mathds R$ respects the order along $\dside{}{w}$} if 
  (i) $i=1$ and $\tau(w) \leq y < \tau(v_2)$
(where equality is only allowed if $\leftcorner{}{v_1}$ does not lie on $\vside{}{w_b}$), 
(ii) $i=2, \dots,\ell-1$ and
$\tau(v_{i-1}) < y < \tau(v_{i+1})$, or (iii) $i=\ell$ and
$\tau(v_{i-1}) < y \leq \ytop(w)$ and $y< \ytop(v_\ell)$.
Further, for a vertex $v$, we consider the ratio $\lambda(v)$ with which $\rightcorner{}{v}$ cuts the incident diagonal side, i.e.,
$   \lambda(v) = \frac{\tau(v) - \tau(v_g)}{\ytop(v_g) - \tau(v_g)}$, if either $v$ is an inner vertex or $v \in \{X_b,X_r\}$,
$\rightcorner{}{v}$ is on $\dside{}{X_g}$, and $v_g:=X_g$.

\begin{figure}[b]
\begin{subfigure}[t]{.5\textwidth}\centering
\includegraphics[page=46]{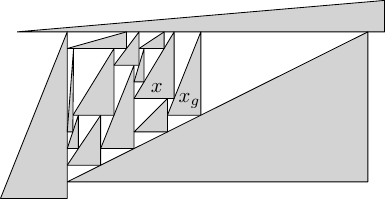}
\subcaption{\label{FIG:before_adjust}original RT-representation}
\end{subfigure}%
\begin{subfigure}[t]{.5\textwidth}\centering
\includegraphics[page=47]{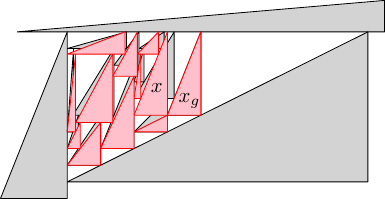}
\subcaption{\label{FIG:after_adjust}after applying \textsc{adjust}$(.,.,x,\tau(x_g))$.}
\end{subfigure}
\caption{\label{FIG:adjust}Moving $\Delta(x)$ down along $\protect\dside{}{x_g}$.}
\end{figure}

For the vertex $x$ and the y-coordinate $y$ that are part of the input of \textsc{adjust}, we assume that moving $x$ to $y$ respects the order along $\dside{}{x_g}$. Setting $\tau(x)\gets y$ may have implications on the neighbors of $x$ of the following type. 
\begin{inparaenum}
\item\label{impl:1} For every vertex $v$ such that $x = v_g$, the value of $\tau(v)$ has to be modified to ensure that the ratio $\lambda(v)$ with which $\rightcorner{}{v}$ cuts $\dside{}{x}$ is maintained; 
\item\label{impl:2} for every vertex $u$ such that $x = u_r$, we have to set $\ytop(u)=y$ to maintain the contact between $\Delta(u)$ and $\Delta(x)$.
\end{inparaenum}
Since these modifications may change the diagonal side of $\Delta(u)$ and $\Delta(v)$, they may trigger analogous implications for the neighbors of $u$ and $v$. 

Since the y-coordinate of $\hside{}{u}$ is not changed, only a type-\ref{impl:1} implication may be triggered for the neighbors of $u$. Further, the two implications correspond to following either a red or a green edge, respectively, in reverse direction with respect to the one in $T$. Hence, the vertices whose triangles may need to be adjusted are those that can be reached from the vertex $x$ by a reversed directed path in $T$ using only red and green edges, but no two consecutive red edges; see \cref{FIG:unchanged_cycle}. Note that, since the green and the red edges have opposite orientation in $T$ and in $\DAG_b(T)$, which is acyclic, this implies that \textsc{adjust} terminates. 

The procedure {\sc adjust} (see \Cref{ALGO:adjust} in the Appendix for its pseudo-code and \cref{FIG:adjust} for an illustration) first finds all the triangles that may need to be adjusted, by performing a simple graph search from $x$ following the above described paths of red and green edges. In a second pass, it performs the adjustment of each triangle $\Delta(w)$, by modifying $\tau(w)$ so that $\lambda(w)$ is maintained. We ensure that the new value of $\tau(w)$ is computed only after the triangle $\Delta(w_g)$ has already been adjusted. 

\begin{restatable}{lemma}{adjust}\label{le:adjust}
  Let $R_1$ be an RT-representation of a plane triangulation $G=(V,E)$ corresponding to the Schnyder wood $T$ and let the y-coordinate of
  $\hside{1}{v}$ be $\tau_1(v)$, $v \in V$. Let $x \in V$ be an inner
  vertex and let $y \in \mathds R$ be such that moving $x$ to $y$ respects the
  order along $\dside{}{x_g}$. Let $\tau_2$ be the output of
  $\textsc{adjust}(\tau_1,T,x,y)$. 

Then, we have that
\begin{inparaenum}[$(i)$]
  \item \label{ITEM:tau_xisy} $\tau_2(x)=y$,
  \item \label{ITEM:ratio_preserved} $\lambda(v)$ is maintained for any vertex $v \neq x$,
  \item \label{ITEM:ADT_labeling}$\tau_2 $ is an ADT-labeling of $DAG_r(T)$, and
  \item \label{ITEM:linear_morph}the morph between $R_1$ and $R_2=\RT(T, \tau_2, R_o)$ is linear, where $R_o=\Delta_1(X_b) \cup \Delta_1(X_g) \cup \Delta_1(X_r)$.
\end{inparaenum}
\end{restatable}

\begin{figure}[t]
  \centering
  \includegraphics[page=45]{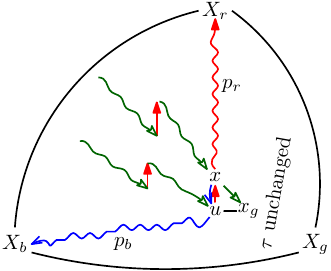}
  \caption{\label{FIG:unchanged_cycle} Let $u$ be the neighbor of $x$
    following $x_g$ in clockwise order. Then, the cycle $C'$ composed
    of the blue $u$-$X_b$-path $p_b$, the edge $\{u,x\}$, the red
    $x$-$X_r$-path $p_r$, and the edge $\{X_r,X_b\}$ encloses all
    vertices for which $\tau$ is changed. Vertex $x$ is the only
    vertex on $C'$ for which $\tau$ is changed.}
\end{figure}

\paragraph{Outline of the proof}.
  \cref{ITEM:tau_xisy,ITEM:ratio_preserved} are clear from
  the construction.  We establish a cycle $C'$ (see
  \cref{FIG:unchanged_cycle}) that encloses all vertices for which
  $\tau$ might be changed.  Distinguishing the cases $y<\tau_1(x)$ and
  $y>\tau_1(x)$, \cref{ITEM:ADT_labeling} can be shown by induction on
  a suitable ordering of the edges in $\DAG_r(T)$. Since all
  predecessors of $x_g$ in $\DAG_b(T)$ are outside or on $C'$, we have that $\Delta_1(x_g)=\Delta_2(x_g)$. Now,
  by \cref{THEO:charRTmorph}, $\left<R_1,R_2\right>$ is a
  linear morph.  
  \qedhere

\subsection{A Flipping Algorithm}\label{SEC:flipping-algorithm}
Recall that, given a Schnyder wood $T$ and an oriented cycle $C$ in $T$, the Schnyder wood $T_C$ is obtained from $T$ by flipping $C$. In the following theorem we show how to realize this flip geometrically with two linear morphs in the case in which $C$ is a facial cycle.
\begin{theorem}\label{THEO:adjustCCC}
  Let $R_1$ be a non-degenerate RT-representation of a plane triangulation
  $G$ corresponding to a Schnyder wood $T$. Let $C$ be an oriented
  facial cycle in $T$.  We can construct a sequence of two linear morphs
  $\left<R_1,R_2,R_3\right>$ such that $R_3$ is a non-degenerate
  RT-representation of $G$ corresponding to a Schnyder wood $T_C$.
\end{theorem}
\begin{proof}
  For the oriented facial cycle $C$, let $C_r$, $C_g$, and $C_b$
  be the vertices with outgoing red, green, and blue
  edge, respectively, in $C$. In order to flip $C$, we move $C_g$
  along the respective incident diagonal sides as sketched in the
  following figure.
  \begin{center}
    \includegraphics[page=48]{figures}
    \end{center}
    More precisely, let $\tau_1$ be the y-coordinates of the
  horizontal sides in $R_1$.
  We first compute $\tau_2\gets\textsc{adjust}(\tau_1,T,C_g,\tau_1(C_b))$.
  If $C$ is clockwise oriented, we then compute
  \begin{quote}
    $\tau_3\gets\textsc{adjust}(\tau_2,T_C,C_g,(\tau_2(C_g) + 
    \max\{\tau_2(u);\; u=C_r \textup{ or } u_g = C_r\})/2)$.
\end{quote}
If $C$ is counter-clockwise oriented, we
proceed as follows.
\begin{quote}
    $\tau_3\gets\textsc{adjust}(\tau_2,T_C,C_g,(\tau_2(C_g) + \min\{\tau_2(u);\; u=(C_b)_r \textup{ or } u_g = C_b\})/2)$.
  \end{quote}
  In each case the new y-coordinates $y$ for $C_g$ are chosen such that moving $C_g$ to $y$ respects the order along the respective incident diagonal. Thus, \textsc{adjust} can be applied. Also, $\tau_2$ is an ADT-labeling of both, $\DAG_r(T)$ and $\DAG_r(T_C)$, and $\tau_3$ is an ADT-labeling of $\DAG_r(T_C)$.  Let $R_2=\RT(T, \tau_2, R_o)=\RT(T_C, \tau_2, R_o)$ and let $R_3=\RT(T_C, \tau_3, R_o)$. Since $\tau_2$ and $\tau_3$ are produced by \textsc{adjust}, by \cref{le:adjust}, both
  $\left<R_1,R_2\right>$ and $\left<R_2,R_3\right>$ are linear~morphs.
\end{proof}

\subsection{Morphing Representations with the same Schnyder Wood}
\label{SEC:same-wood-different-proportions}

In this section, we consider RT-representations corresponding to the
same Schnyder Wood.

\begin{theorem}\label{THEO:adjustCC}
  Let $R_1$ and $R_2$ be two RT-representations of an $n$-vertex plane
  triangulation corresponding to the same Schnyder wood
  $T$. Then, there is a piecewise linear morph between
  $R_1$ and $R_2$ of length at most $2n$.
\end{theorem}

The idea is to first transform the outer face to a canonical
form, and then to move one vertex $v$ per step to a new y-coordinate $y$
such that the ratio $\lambda(v)$ is set to how it should be in
$R_2$. The order in which we process the vertices is such
that $\textsc{adjust}$ can be applied to the vertex $v$ and the
y-coordinate $y$.  Recall that \textsc{adjust} does not
alter the ratio $\lambda$, except for the currently processed vertex $v$.
The following lemma can be proven by induction on $n$.

\begin{restatable}{lemma}{ordering}
  Let $P=\{p_1 < \dots < p_n\}$ and $Q=\{q_1 < \dots < q_n\}$ be two
  sets of $n$ reals each. If $P \neq Q$ then there is an $i$ such that
  $p_i \neq q_i$ and $P$ has no element between $p_i$ and
  $q_i$.
\end{restatable}
\begin{corollary}
  Let $P$ and $Q$ each be a set of $n$ points on a segment
  $s$. We can move $P$ to $Q$ in $n$ steps by moving one point per
  step and by maintaining the ordering of the points on $s$.
\end{corollary}

\paragraph{Proof of \cref{THEO:adjustCC}.}
  Let $\tau'$ be a topological ordering of the inner vertices of
  $\DAG_r(T)$. We extend $\tau'$ to an ADT-labeling of $\DAG_r(T)$ by setting $\tau'(X_b)=0=\tau'(X_g)$, and 
  $\tau'(X_r)=n-2$. With a sequence of at most $n$
  linear morphs we transform $R_i$, $i=1,2$, into an
  RT-representation $R'=\RT(T,\tau',R_o)$, where $R_o$ has the following \emph{canonical form}:
  $\rightcorner{}{X_b} = \leftcorner{}{X_g} = (0,0)$,
  $\topcorner{}{X_b} = \leftcorner{}{X_r} = (0,n-2)$,
  $\topcorner{}{X_g} = \rightcorner{}{X_r} = (n-2,n-2)$, and the
  lengths of $\vside{}{X_r}$ and $\hside{}{X_b}$ are one.
  In the first morph, we cut the extruding parts of the outer
  triangles. In the second morph, we independently scale the x- and
  y-coordinates of the corners and translate the drawing, to
  fit the corners as indicated. In a third step, we adjust the lengths
  of $\vside{}{X_r}$ and $\hside{}{X_b}$.
  In the first morph the slope of no side is changed, in
  the second morph no ratio is changed, and in the third morph there
  are only four sides that are changed, which are not incident to
  any other triangle. Thus, the three morphs are linear.
  Let the resulting RT-representation be $R_i'$.
  
  We now process the vertices in a reversed topological ordering
  on $\DAG_b(T)$.
  We process a vertex $w$ as follows. Let $\tau$ be the current y-coordinates of the horizontal sides. Let $\mathcal G(w) = \{v \in V; w=v_g\}$ and let
  $P = \{\tau(v);\;v \in \mathcal G(w)\}$.  For
  $v \in \mathcal G(w)$ let $y(v)$ be such that

  \[
  \frac{y(v)-\tau(w)}{\tau(w_r)-\tau(w)}=\frac{\tau'(v)-\tau'(w)}{\tau'(w_r)-\tau'(w)},
  \]
  i.e., placing $\hside{}{v}$, $v \in \mathcal G(w)$ on the
  y-coordinate $y(v)$ cuts $\dside{}{w}$ in the the same ratio as
  in $R'$.  Let $Q=\{y(v);\;v \in \mathcal G(w)\}.$ By the above
  corollary, we can order $\mathcal G(w)=\{v_1,\dots,v_k\}$ such that
  replacing in the i$th$ step $\tau(v_i)$ by $y(v_i)$ maintains the
  ordering of $\{\tau(v);\;v \in \mathcal G(w)\}$. Since $\tau'$ is a
  topological ordering, we will not move $\rightcorner{}{v_i}$ to an
  end vertex of $\dside{}{w}$.
  For $i=1,\dots,k$ we now call
  $\tau \gets \textsc{adjust}(\tau,T,v_i,y(v_i))$. This yields one
  linear morphing step.

  After processing all vertices $w$ in a reversed topological ordering
  of $\DAG_b(T)$ and all vertices in $\mathcal G(w)$ in the order given
  above, we have obtained an RT-representation $R$ in which any right
  corner cuts its incident diagonal in the same ratio as in
  $R'$. Since the outer face is fixed, this implies that $R=R'$.
  Observe that $\mathcal G(w)$, $w \in V$, is a partition of the set of
  inner vertices. Hence, we get at most one morphing step for each
  of the $n-3$ inner vertices.
  \qedhere
  \smallskip


Combining the results of \cref{SEC:moving_a_triangle,SEC:flipping-algorithm,SEC:same-wood-different-proportions} yields the main result of the section.

\paragraph{Proof of \cref{THEO:morph_exists}.}
  First, we transform $R_1$ into a non-degenerate RT-representation $R$
  with Schnyder wood $T_1$ and a canonical representation of the outer
  face in $\mathcal O(n)$ linear morphing steps, by \cref{THEO:adjustCC}. Then, we perform the
  $\ell$ facial flips as described in the proof of
  \cref{THEO:adjustCCC}, using two linear morphs for each
  flip. This yields an RT-representation $R'$ with Schnyder wood
  $T_2$. Finally, we transform $R'$ into $R_2$ in $\mathcal O(n)$
  linear morphing steps, by \cref{THEO:adjustCC}. This yields a total
  of $\mathcal O (n+\ell)$ linear morphs. Each linear morph can be
  computed by one application of $\textsc{adjust}$, which runs in
  linear time.
  \qedhere

\section{A Decision Algorithm}\label{SEC:decision_algorithm}

It follows from \cref{THEO:nec} and \cref{THEO:morph_exists} that
there is a piecewise linear morph between two RT-representations of
a plane triangulation if and only if the respective Schnyder woods can
be obtained from each other by flipping faces only. Note that this condition is always 
satisfied if the triangulation is 4-connected and the topmost vertex is the
same in both RT-representations. On the other hand, if the graph contains
separating triangles, we have to decide whether there is such a
sequence of facial flips. We will show that this can
be decided efficiently and that, in the positive case, there exists one sequence whose length is
at most quadratic in the number of vertices. This establishes our
final result.

\begin{theorem}\label{THEO:main}
  Let $R_1$ and $R_2$ be two RT-representations of an $n$-vertex plane
  triangulation. We can decide in
  $\mathcal O(n^2)$ time whether there is a piecewise linear morph
  between $R_1$ and $R_2$ and, if so, a morph with $\mathcal O(n^2)$ linear
  morphing steps can be computed in $\mathcal O(n^3)$ time.
\end{theorem}

Since there is a one-to-one correspondence between Schnyder woods of a
plane triangulation and its 3-orientations, we will omit the colors in the
following.
%
A careful reading of Brehm~\cite{brehm:diplom} and
Felsner~\cite{felsner:04} reveals the subsequent properties of
3-orientations.  The set of 3-orientations of a triangulation forms a
distributive lattice with respect to the following ordering.  $T_1 \leq T_2$
if and only if $T_1$ can be obtained from $T_2$ by a sequence of flips
on some counter-clockwise triangles. The minimum element is the unique
3-orientation without counter-clockwise cycles. Moreover, given a
3-orientation $T$ and a triangle $t$, the number of occurrences of $t$
in any flip-sequence between $T$ and the minimum 3-orientation is the
same~--~provided that the flip sequence contains only
counter-clockwise triangles. Let this number be the potential~$\pi_T(t)$.
See \cref{FIG:lattice} in the Appendix for an example.

Observe that $\pi_T$ is distinct for distinct $T$. Moreover,
$\min(\pi_{T_1}(t),\pi_{T_2}(t))$, $t$ triangle, is the
potential of the \emph{meet} $T_1 \wedge T_2$ (i.e., the infimum) of two 3-orientations $T_1$
and $T_2$, while $\max(\pi_{T_1}(t),\pi_{T_2}(t))$, $t$ triangle,
is the potential of the \emph{join} $T_1 \vee T_2$ (i.e., the supremum) of $T_1$ and $T_2$.
The potential $\pi_{T}$ can be computed in quadratic
time for a fixed 3-orientation $T$ of an $n$-vertex triangulation: At
most $\mathcal O(n^2)$ flips have to be performed in order to reach
the minimum 3-orientation. With a linear-time preprocessing, we can
store all initial counter-clockwise triangles in a list. After each
flip, the list can be updated in constant time.

\begin{lemma}
  Let $T_1$ and $T_2$ be two 3-orientations of an $n$-vertex
  triangulation. $T_1$ can be obtained from $T_2$ by a sequence of
  facial flips if and only if $\pi_{T_1}(t) - \pi_{T_2}(t) = 0$
  for all separating triangles $t$. Moreover, if $T_1$ can be obtained
  from $T_2$ by a sequence of facial flips, then it can be obtained by
  $\mathcal O(n^2)$ facial flips.
\end{lemma}
\begin{proof}
  Observe that going from $T_1$ to the meet $T_1 \wedge T_2$ involves
  \[\pi_{T_1}(t) - \min(\pi_{T_1}(t),\pi_{T_2}(t)) \in
  \{0,\pi_{T_1}(t)-\pi_{T_2}(t)\}\] counter-clockwise flips on
  triangle $t$, and going from the meet $T_1 \wedge T_2$ to $T_2$
  involves
  \[\pi_{T_2}(t) - \min(\pi_{T_1}(t),\pi_{T_2}(t)) \in
  \{0,\pi_{T_2}(t)-\pi_{T_1}(t)\}\] clockwise flips on triangle
  $t$. Thus, if $\pi_{T_1}(t) - \pi_{T_2}(t) = 0$ for all
  separating triangles $t$, then no flip must be performed on a
  separating triangle. Then, the total number of flips is bounded by
  $\sum_{t \textup{ face}}(\pi_{T_1}(t) + \pi_{T_2}(t)) \in
  \mathcal O(n^2)$.

  Assume now that there is a sequence
  $T_1=T_0',T_1',\dots,T_\ell',T_{\ell+1}'=T_2$ of 3-orientations such
  that $T_{i+1}'$, $i=0,\dots,\ell$, is obtained from $T'_i$ by a
  (clockwise or counter-clockwise) facial flip. We show by induction on
  $\ell$ that $\pi_{T_1}(t) - \pi_{T_2}(t) = 0$ for all
  separating triangles $t$. If $\ell=0$, let $t_0$ be the triangle that
  has to be flipped in order to go from $T_1$ to $T_2$. Then, $t_0$
  is a face and
  $\pi_{T_1}(t) - \pi_{T_2}(t) = 0$ for $t \neq t_0$.
  Assume now that $\ell\geq 1$. Let $t$ be a separating triangle. Then
  \begin{align*}
    \pi_{T_1}(t) - \pi_{T_2}(t) &= \underbrace{\pi_{T_1}(t) - \pi_{T'_\ell}(t)}_{=0 \textup{ by IH}} + \underbrace{\pi_{T'_\ell}(t) - \pi_{T_2}(t)}_{=0 \textup{ by IH}} = 0. 
  \end{align*}
\end{proof}

Observe that there might be a piecewise linear morph between two RT-representations even though a separating triangle in the respective Schnyder woods is oriented in opposite directions. E.g., consider the 3-orientations III and II of the graph in \cref{FIG:lattice}. Moreover, there might be two RT-representations such that any separating triangle is oriented in the same way in the two respective Schnyder woods, however, there is no piecewise linear morph between the two representations. E.g., consider the 3-orientations III and I of the graph in \cref{FIG:lattice}.

\section{Conclusions and Open Problems}

We have studied piecewise linear morphs between RT-representations of plane triangulations, and shown that when such a morph exists, there is one of length $\mathcal{O}(n^2)$. 
It would be interesting to explore lower bounds on this length. Observe that the minimum length of a flip-sequence containing only facial cycles does not immediately imply such bound, since some flips could be parallelized.
Additionally, bounds on the resolution throughout our morphs would be worth investigating; however, it is unclear whether the ``ratio fixing'' we use would allow nice bounds. For this, it may help to return to integer y-coordinates between any two flips; however, this would result in a cubic number of linear morphing steps.
A major open direction is whether our results can be lifted to general plane graphs, e.g., through the use of compatible triangulations. Note that such a compatible triangulation would need to be formed while preserving the conditions for the existence of a linear morph, i.e., without introducing the need to flip a separating triangle.

Finally, beyond the context of RT-representations, many other families of geometric objects could be considered. 
For example, morphing degenerate contact representations of line segments generalizes planar morphing, by treating contact points as vertices. 

\paragraph{Acknowledgements.}
This research began at the Graph and Network Visualization Workshop 2018 (GNV'18) in Heiligkreuztal.
We thank Stefan Felsner, Niklas Heinsohn, and Anna Lubiw for~interesting~discussions on this subject.

\bibliographystyle{alpha} 
\bibliography{../biblio}

\clearpage

\appendix
\section{Appendix}

In this appendix we present the full version of the proofs that have been omitted or sketched in the paper.

\threeInACorner*

\begin{proof}
Let $f_\ell$ and $f_r$ be the two triangular faces to the left and the right of $\{u,v\}$, and let $w_\ell$ and $w_r$ be the vertex incident to $f_\ell$ and $f_r$, respectively, that is different from $u$ and $v$. Then, $\Delta(w_\ell)$ and $\Delta(w_r)$ must both touch $\Delta(u)$ and $\Delta(v)$, and they must be on different sides of $p$, unless one of $f_\ell$ and $f_r$ is the outer face. This implies that either $\Delta(w_\ell)$ or $\Delta(w_r)$ has a corner on $p$.
\end{proof}

\fromTandTautoR*

\begin{proof} 
  Let $\tau': V \leftrightarrow \{1,2,\dots,n\}$ be a topological ordering of $\DAG_r(T)$. We process the vertices of $G$ according to $\tau'$. By the construction of $\DAG_r(T)$, we have that either $\tau'(X_b)=1$ and $\tau'(X_g)=2$ or $\tau'(X_g)=1$ and $\tau'(X_b)=2$.  Thus, in the first two steps, we draw triangles $\Delta(X_b)$ and $\Delta(X_g)$ as in $R_o$; see  \cref{FIG:schnyder-wood-to-representation}. At each of the following steps, we consider a vertex $v$, with $\tau'(v)=i > 2$. Let $R_{i-1}$ be the RT-representation of the subgraph $G_{i-1}$ of $G$ induced by the vertices preceding $v$ in $\tau'$. 
  
  We first consider the case in which $v \neq X_r$, that is, $i < n$. We draw $\Delta(v)$ with its horizontal side on $y=\tau(v)$ and with its top corner at $y=\tau(v_r)$, as follows. Since the blue edge $(v_b,v)$ and the green edge $(v_g,v)$ are entering $v$ in $\DAG_r(T)$, we have that $\tau'(v_b)<\tau'(v)$ and $\tau'(v_g)<\tau'(v)$. Therefore, triangles $\Delta(v_b)$ and $\Delta(v_g)$ have already been drawn in $R_{i-1}$. Also, by \cref{cond:tau-1} of ADT-labeling, we have that $\tau(v_b)\leq \tau(v)$ and $\tau(v_g) \leq \tau(v)$.
  
  We show that the top corners of $\Delta(v_b)$ and $\Delta(v_g)$ have y-coordinate larger than or equal to $\tau(v)$. We present our proof for $\Delta(v_b)$, the arguments for $\Delta(v_g)$ are analogous. If $v_b = X_b$, this follows from the fact that $\topcorner{}{X_b}$ has y-coordinate larger than or equal to $\tau(X_r)$ in $R_o$. Otherwise, recall that the y-coordinate of  $\topcorner{}{v_b}$ has been set equal to $\tau(r)$, where $r$ is the neighbor of $v_b$ such that edge $(v_b,r)$ of $\DAG_r(T)$~is~red. 
  
  Let $w_1, \dots, w_k$ (with $w_1=v$ and $w_k=r$) be the neighbors of $v_b$ such that edges $(v_b, w_1), \dots, (v_b, w_k)$ appear in this counter-clockwise order around $v_b$. Since $(v_b, w_1)$ is blue and entering $v_b$ in $T$, while $(v_b, w_k)$ is red and exiting $v_b$ in $T$, we have that each edge $(v_b, w_i)$, with $2 \leq i \leq k-1$, is blue and entering $v_b$~in~$T$. 
  
  We claim that each edge $(w_j, w_{j+1})$, with $j=1, \dots, k-1$, is either red and exiting $w_j$ in $T$ or green and entering $w_j$ in $T$. Observe that if the claim holds, then $w_1,\dots,w_k$ form a directed path from $v_b$ to $r$ in $\DAG_r(T)$, and hence $\tau(r) \geq \tau(v)$. Thus, $\topcorner{}{v_b}$ has y-coordinate larger than or equal to $\tau(v)$, as~desired.
  
  To prove the claim, first consider the case $j \leq k-2$. Since $(v_b, w_j)$ and $(v_b, w_{j+1})$ are both blue, edge $(w_j, w_{j+1})$ cannot be blue. Further, if it was red and entering $w_j$ in $T$, we would have that the blue and the red edges exiting $w_{j+1}$ in $T$ are consecutive in counter-clockwise order around $w_{j+1}$, which is not possible. Analogously, if $(w_j, w_{j+1})$ was green and exiting $w_j$ in $T$, we would have that the blue and the green edges exiting $w_{j}$ in $T$ are consecutive in clockwise order around $w_j$, which is not possible. Thus, the claim follows in this case.
  Consider the case $j=k-1$. Since edge $(v_b,w_k)$ is red and entering $w_k$ in $T$, the edge $(w_{k-1},w_k)$, which follows $(v_b,w_k)$ in the counter-clockwise order around $w_k$, can be either red and entering $w_k$ in $T$, or green and exiting $w_k$ in $T$, and the claim follows.
  
  Further, it can be seen that our construction maintains the invariant that, if the vertices $v_1=X_b,\dots,v_q=X_g$ forming a path along the outer face of $G_{i-1}$, then the top corners of $\topcorner{}{v_1},\dots,\topcorner{}{v_q}$ have increasing x-coordinates in $R_{i-1}$. This implies that, if a triangle $\Delta(u)$ intersects the line $y=\tau(v)$ between $\Delta(v_b)$ and $\Delta(v_g)$, then $u$ is a neighbor of $v$ such that edge $(u,v)$ is red and oriented from $u$ to $v$ in $\DAG_r(T)$, that is, $v = u_r$. By construction, $\topcorner{}{u}$ has y-coordinate equal to $\tau(v)$. 
  
  Thus, by drawing the horizontal side of $\Delta(v)$ on $y=\tau(v)$ between $\Delta(v_b)$ and $\Delta(v_g)$, we obtain a representation $R_i$ of the graph $G_i=G_{i-1} \cup v$ in which $\Delta(v)$ touches $\Delta(v_b)$ with pair $(\leftcorner{}{v},\vside{}{v_b})$, it touches $\Delta(v_g)$ with pair $(\rightcorner{}{v},\dside{}{v_g})$, and it touches each vertex $u$ such that $v = u_r$ with pair $(\topcorner{}{u},\hside{}{v})$, and does not touch any other triangle. Note that these are exactly the compatible pairs that are enforced by $T$.
  
  Thus, to show that $R_i$ is an RT-representation, it remains to prove that the horizontal side of $\Delta(v)$ has positive length. This is clear if $\tau(v)$ is strictly greater than both $\tau(v_g)$ and $\tau(v_b)$. So assume first that $\tau(v_b)<\tau(v) = \tau(v_g)$. Then, by \cref{cond:tau-2} of ADT-labeling, there is a vertex $u$ such that $\left<v,v_g,u\right>$ is a clockwise facial cycle. It follows that $(u,v)$ is an incoming red edge of $v$ in $T$ and thus $\topcorner{}{u}$ lies between $\Delta(v_b)$ and $\Delta(v_g)$ on the horizontal line $y=\tau(v)$, as discussed above. Since the horizontal side of $\Delta(u)$ has a positive length in $R_{i-1}$ and since $\tau(v_b) < \tau(v)$, the vertical side of $\Delta(v_b)$ must be strictly to the left of $\topcorner{}{u}$, and thus strictly to the left of $\leftcorner{}{v_g}$. The case $\tau(v_b) = \tau(v) > \tau(v_g)$ is analogous.
  Finally, we consider the case $\tau(v_b) = \tau(v) = \tau(v_g)$. By \cref{cond:tau-3} of ADT-labeling, there exist two distinct vertices $u_1$ and $u_2$ such that $\left<v,v_g,u_1\right>$ is a clockwise facial cycle and  $\left<v,v_b,u_2\right>$ is a counter-clockwise facial cycle. Since the horizontal side of $\Delta(u_1)$ and $\Delta(u_2)$ have positive length in $R_{i-1}$, the right corner of $\Delta(v_b)$ and the left corner of $\Delta(v_g)$ cannot coincide.
  
  Therefore $R_i$ is an RT-representation when $i < n$. If $i=n$, and hence $v = X_r$, we draw $\Delta(X_r)$ as in $R_o$. This implies that $\Delta(X_r)$ touches $\Delta(X_g)$ and $\Delta(X_b)$. Since for each neighbor $u$ of $X_r$ different from $X_b$ and $X_g$, edge $(u,X_r)$ is red and entering $X_r$ in $\DAG_r(T)$, and since by construction $\topcorner{}{u}$ has y-coordinate equal to $\tau(X_r)$, we have that $\topcorner{}{u}$ touches $\hside{}{X_r}$. This concludes the proof that $R_n=\RT(T, \tau, R_o)$ is an RT-representation of $G$ satisfying the statement. The uniqueness of $\RT(T, \tau, R_o)$ follows from the fact that the y-coordinates of the horizontal sides are fixed by $\tau$, and for every inner vertex $v$ the y-coordinate of $\topcorner{}{v}$ must be $\tau(v_r)$, and the x-coordinates must be chosen such that $\Delta(v)$ touches both $\Delta(v_b)$ and $\Delta(v_g)$.
\end{proof}

\charRTmorph*

\begin{proof}
  Let $R_t$ be the representation at time instant $t$ of the linear morph from $R_0$ to $R_1$, with $t \in [1, 2]$, i.e., the set of triangles obtained by interpolating the corners of each triangle in $R_0$ to $R_1$. 
  It suffices to show that $R_t$ is an RT-representation of $G$.
  
  We start by proving that, for every vertex $u$ of $G$, the triangle $\Delta_t(u)$ representing $u$ in $R_t$ is the lower-right half of an axis-parallel rectangle with positive area. Consider the triangles $\Delta_1(u)$ and $\Delta_2(u)$
  representing $u$ in $R_0$ and $R_1$, respectively. Since \vside{1}{u} and \vside{2}{u} are vertical segments with \topcorner{i}{u} lying above \rightcorner{i}{u}, with $i = 1, 2$, we have that \vside{t}{u} is a vertical segment with \topcorner{t}{u} lying above \rightcorner{t}{u}. A similar argument applies to prove that \hside{t}{u} is a horizontal segment with \leftcorner{t}{u} lying to the left of \rightcorner{t}{u}. Since all these segments have positive length, the claim follows.
  
  Next, we prove that $R_t$ is an RT-representation of $G$, i.e, we show that any two triangles touch in $R_t$ if and only if the corresponding vertices of $G$ are adjacent, and that no two triangles share more than~one~point.
  
  First, consider two vertices $u$ and $w$ that are adjacent in $G$. We show that $\Delta_t(u)$ and $\Delta_t(w)$ touch in exactly one point. Let $c_1(u)$ with $c\in \{\topcornersymb, \leftcornersymb, \rightcornersymb \}$ and $s_1(w)$ with $s \in \{ \hsidesymb, \vsidesymb, \dsidesymb \}$ be the corner of $\Delta_1(u)$ and the side of $\Delta_1(w)$, respectively, that touch in $R_0$.
  Since $R_0$ and $R_1$ correspond to the same Schnyder wood, and since the triangles of the outer face pairwise touch in their corners, we have that the corner $c_2(u)$ of $\Delta_2(u)$ touches the side $s_2(w)$ of $\Delta_2(w)$ in $R_1$. This, together with \cref{LEMMA:sidewithpoint}, implies that the corner $c_t(u)$ of $\Delta_t(u)$ touches the side $s_t(w)$ of $\Delta_t(w)$ in $R_t$. The fact that no other points are shared between $\Delta_t(u)$ and $\Delta_t(w)$ derives from the possible corner-side pairs $(c,s)$ that may appear in an RT-representation.
  This also implies that $R_t$ induces the same Schnyder wood as $R_0$ and $R_1$.
  
  Second, consider two vertices $u$ and $w$ that are not adjacent in $G$. We show that $\Delta_t(u)$ and $\Delta_t(w)$ do not share any point in $R_t$. We are going to exploit the following property of a Schnyder wood $T$ of $G$~\cite{DBLP:journals/ipl/Angelini17,DBLP:journals/tcs/GiacomoLM16,DBLP:journals/jocg/NollenburgPR16}:
  Let $T_r,T_g,T_b$, respectively, be the subtree of $T$ induced by the red, green, and blue edges, respectively.
  There exists a vertex $z$, possibly $z=w$, such that $T_i$ contains a path $P_{u,z}$ from $u$ to $z$ and $T_j$ contains a path $P_{w,z}$ from $w$ to $z$, for some $i \neq j \in \RGB$.
  We will show that $\Delta_t(u)$ and $\Delta_t(w)$ are in different quadrants of $\rightcorner{t}{z}$. See \cref{FIG:quadrants} for an illustration.
  \begin{figure}
    \centering
    \includegraphics[page=58]{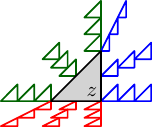}
    \caption{\label{FIG:quadrants}The three quadrants containing triangles on a red, green, or blue path to $z$.}
  \end{figure}

  Consider a unicolored path $P$ to $z$.  Then any corner-side pair
  determined by any edge in $P$ is of the same type, more precisely, of
  type $(\rightcornersymb,\dsidesymb)$ if $P$ is green, of type
  $(\leftcornersymb,\vsidesymb)$ if $P$ is blue, and of type
  $(\topcornersymb,\hsidesymb)$ if $P$ is red.
  Let $a$
  be a vertex. If there is a green $a$-$z$-path then $\Delta_t(a)$ is
  in the upper left quadrant of $\rightcorner{t}{z}$, including the
  boundary above $\topcorner{t}{z}$ or to the left of
  $\leftcorner{t}{z}$. If there is a blue $a$-$z$-path then
  $\Delta_t(a)$ is in the upper right quadrant of $\rightcorner{t}{z}$,
  including the boundary between $\topcorner{t}{z}$ and
  $\rightcorner{t}{z}$. Finally, if there is a red $a$-$z$-path then
  $\Delta_t(a)$ is in the lower left quadrant of
  $\rightcorner{t}{z}$, including the boundary between
  $\leftcorner{t}{z}$ and $\rightcorner{t}{z}$. The only three points
  where triangles of these three regions could intersect are the three
  corners of $\Delta_t(z)$. Assume that a corner of $\Delta_t(u)$ and a corner
  of $\Delta_t(w)$ coincide with the same corner of $\Delta_t(z)$. But
  that would imply that the same corners would also coincide in $R_0$
  and $R_1$~--~contradicting that $u$ and $w$ were not adjacent.
\end{proof}

\ledifferentflips*

\begin{figure}
  \centering
    \includegraphics[page=50]{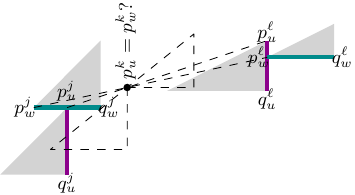}
  \caption{\label{FIG:linearMorphSameWood}A linear morph requires the same Schnyder wood.}
\end{figure}
\begin{proof}
  Assume that there is no Schnyder wood that corresponds to both $R_0$ and $R_1$. Then there must be an edge $\{u,w\}\in E$ such that $\Delta(u)$ would have to go ``around a corner'' of $\Delta(w)$ during a morph. To be more precise, for each $v \in V$, let $\Delta_t(v)$, $t \in [0,1]$, be the triangle representing $v$ at time point $t$.
  Then there must be $0 < j < k < \ell < 1$, a side
  $s^t_u=\overline{p^t_u,q^t_u}$ of $\Delta_t(u)$, and a side
  $s^t_w=\overline{p^t_w,q^t_w}$ of $\Delta_t(w)$ such that $p^j_u$ is
  an interior point of $s^j_w$, $p^k_w = p^k_u$, and $p^\ell_w$ is an
  interior point of $s^\ell_u$.  See~\cref{FIG:linearMorphSameWood}. If the trajectories of $p^t_u$
  and $p^t_w$ do not intersect then $p^k_u=p^k_w$ is impossible.  If
  the trajectories of $p^t_u$ and $p^t_w$ are collinear then
  $p_w^\ell$ cannot be an interior point of $s_u^\ell$.  So assume
  that they intersect in a point $p$. Since $s^j_u$ and $s^\ell_w$ are
  not parallel we have by the intercept theorem that
  $|\overline{p^j_wp}|/|\overline{pp^\ell_w}|\neq
  |\overline{p^j_up}|/|\overline{pp^\ell_u}|$. Thus, the trajectories
  of $p^t_u$ and $p^t_w$ do not pass at the same time $k$ through $p$.
\end{proof}

  \begin{algorithm}
  \Input{RT-representation $R$ of a plane triangulation $G=(V,E)$ with
    $\hside{}{v}$ on $\tau(v)$, $v \in V$, and Schnyder wood $T$. A vertex
    $x \in V$, a real number $y \in \mathds R$ such that moving $x$ to $y$
    respects the order along $\dside{}{x_g}$.}
  \Output{ADT-labeling $\tau$ of $\DAG_r(T)$ with
    (i) $\tau(x)=y$,
    (ii) the ratio $\lambda(v)$ is maintained for any vertex $v \neq x$, and
    (iii) the triangle $\Delta(x_g)$ is the same in $R$ and $\RT(T, \tau, R_o)$, where $R_o$ is the representation of~the~outer~face~in~$R$.}

  \Data{Stack $S$; Boolean $v.\textsc{red}, v.\textsc{green}$, for each $v\in V$ initialized to \textsc{false}, describing whether $v$ was found by traversing a red edge, or a green edge, or both (\textsc{true} in the first round and \textsc{false} in the second).}
\smallskip
  \Los{\textsc{adjust}(\textit{y-coordinates} $\tau$, \textit{Schnyder wood} $T$, \textit{vertex} $x$, \textit{real} $y$)}{ 
    \For{$u \in V$}{
      $\lambda(u) \gets (\tau(u) - \tau(u_g))/(\ytop(u_g) - \tau(u_g))$\;
    }
    $\tau(x) \gets y$\;
    \For{$i=1,2$}{
    $S.\textsc{push}(x)$\;
    \For{each incoming red edge $(u,x)$ of $x$ 
    }{
      $u.\textsc{red} \gets \neg u.\textsc{red}$\tcc*{\textsc{true} in first pass, \textsc{false} in second}
      \If(\tcc*[f]{always fulfilled in first pass}){$\neg u.\textsc{green}$}{
        $S.\textsc{push}(u)$\; 
      }
    }
    \While{$S \neq \emptyset$}{
      $w \gets S.\textsc{pop}$\; 
      \For{each incoming green edge $(v,w)$ of $w$
      }{
        \If{$i=2$}{
          $\tau(v) \gets \lambda(v)\cdot(\tau(w_r) - \tau(w)) + \tau(w)$\tcc*{$\tau(w_r) = \ytop(w)$}
        }
        $v.\textsc{green} \gets \neg v.\textsc{green}$\;
                \If{$\neg v.\textsc{red}$}{
          $S.\textsc{push}(v)$\;
        }
        \For{each incoming red edge $(u,v)$ of $v$
        }{
          $u.\textsc{red} \gets \neg u.\textsc{red}$\;
          \If{$\neg u.\textsc{green}$}{
            $S.\textsc{push}(u)$\; 
          }        
        }
      }
    }
  }
}
\caption{Pseudo-code of procedure {\sc adjust}}\label{ALGO:adjust}
\end{algorithm}

\adjust*

\begin{proof}
  The fact that $\tau_2(x)=y$ and that $\lambda(v)$ is maintained for any vertex $v \neq x$ follows by construction. 
To prove that $\tau_2$ is an ADT-labeling, we first determine a cycle bounding a subgraph of $G$ containing all the vertices whose $\tau$ might be modified by procedure \textsc{adjust}.
  \begin{description}
  \item[A cycle $C'$ enclosing all vertices where $\tau$ might
    change:] Let $u$ be the neighbor of $x$ following $x_g$ in clockwise order around $x$. Observe that either $x =u_r$ or $u=x_b$.     
    Consider the     blue $u$-$X_b$-path $p_b$ and the red $x$-$X_r$-path $p_r$ of $G$ (see
    \cref{FIG:unchanged_cycle}). Let $C'$ be the cycle composed by
    $X_b,p_b,u,x,p_r,X_r$. Let $v \neq x$ be a vertex along or outside
    $C'$: We claim that $\tau_2(v)=\tau_1(v)$. In fact, $\tau_2(v)\neq \tau_1(v)$ only if 
  there exists in $\DAG_b(T)$ a directed red-green-path $p$ with no two consecutive red edges starting at $x$ and ending with a green edge at $v$. This is due to the fact that, when processing a vertex $w$, procedure \textsc{adjust} might choose to process a neighbor $v$ of $w$ only if $T$ contains the green edge $(v, w)$ and then it might choose to process a neighbor $u$ of $v$ only if $T$ contains the red edge $(u,v)$. Recall that green and red edges have opposite orientation in $T$ and $\DAG_b(T)$.
  
  Note that $p$ starts at $x$, which lies along $C'$ and, after possibly reaching $u$ via a red edge, it necessarily enters the interior of $C'$ via a green edge. In the following steps $p$ can reach again a vertex $z$ of $C'$; however this can only happen if $z$ belongs to $p_b$ and it is reached via a red edge. This implies that the next edge of $p$ is green and hence $p$ enters again the interior of $C'$.
\item[$\tau_2$ is an ADT-labeling]
  \begin{figure}
    \centering
    \begin{tabular}{cccc}
      \begin{subfigure}{.22\textwidth}
        \centering
        \includegraphics[page=51]{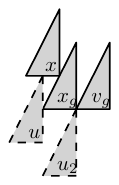}
        \subcaption{$x=v_b$, $x_g=v$}
      \end{subfigure}&
      \begin{subfigure}{.22\textwidth}
        \centering
        \includegraphics[page=52]{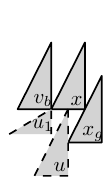}
        \subcaption{$x=v$, $x_g=v_g$}
      \end{subfigure}&
      \begin{subfigure}{.22\textwidth}
        \centering
        \includegraphics[page=53]{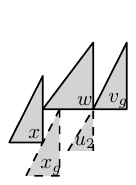}
        \subcaption{$x=v_b$, $w=v$}
      \end{subfigure}&
      \begin{subfigure}{.22\textwidth}
        \centering
        \includegraphics[page=54]{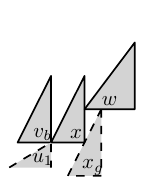}
        \subcaption{$x=v$, $w=v_g$}
      \end{subfigure}\\
      \multicolumn{2}{c}{$\left<x,x_g,u\right>$ clockwise}&
      \multicolumn{2}{c}{$\left<w,x,x_g\right>$ counter-clockwise}
    \end{tabular}
    \caption{\label{FIG:correctness_ADT3}If the horizontal sides of
      the solid triangles are on one line in $R_2$ then the two
      dashed triangles closing a counter-clockwise facial cycle with $(v,v_b)$
      and a clockwise facial cycle with $(v,v_g)$, respectively, are
      distinct. Situation drawn with respect to $R_1$.}
  \end{figure}
  If $y=\tau_1(x)$ then $\tau_1 = \tau_2$ is an ADT-labeling. Otherwise,
  distinguishing the cases $y<\tau_1(x)$ and $y>\tau_1(x)$
  we will prove by induction on a suitable ordering of the edges in
  $\DAG_r(T)$ the following property.
  \begin{description}
  \item[Property DTL:] $\tau_2(v) \leq \tau_2(w)$ for any edge $(v,w)$
    of $\DAG_r(T)$ where equality may only hold if (i)
    $\tau_1(v) = \tau_1(w)$, (ii) $w=x$, $v=x_g$, and
    $\left<w,v,v_b\right>$ is a clockwise oriented facial cycle in
    $T$, or (iii) $v=x=w_b$ and $\left<w,v,v_g\right>$ is a
    counter-clockwise oriented facial cycle in~$T$.
  \end{description}
  Since $\tau_1$ is an ADT-labeling of $\DAG_r(T)$ this
  immediately establishes Properties 1+2 of an ADT-labeling for
  $\tau_2$. In order to prove Property~3, assume that
  $\tau_2(v_g)=\tau_2(v)=\tau_2(v_b)$ for an inner vertex $v$. If
  $x \notin \{v_b,v\}$ then Property DTL implies
  $\tau_1(v_g)=\tau_1(v)=\tau_1(v_b)$ and thus, Property~3 is
  fulfilled.  So assume that $x \in \{v_b,v\}$. We distinguish the two
  cases (ii) and (iii) of Property DTL (see \cref{FIG:correctness_ADT3}):

  Assume first there is a clockwise facial cycle
  $\left<x,x_g,u\right>$ implying that $(x_g,u)$ is red, $(u,x)$
  is blue, and $y = \tau_2(x)=\tau_2(x_g)=\tau_1(x_g)$. Now, if
  $x=v_b$ then $v=x_g$ and
  $\tau_1(x_g) = \tau_2(x_g) = \tau_2(v_g) =\tau_1(v)$ implies that
  there is a vertex $u_2$ such that
  $\topcorner{1}{u_2} = \rightcorner{1}{x_g} = \leftcorner{1}{v_g}\neq
  \topcorner{1}{u}$. Thus $u \neq u_2$ close the two facial cycles
  with $(v_b,v)$ and $(v,v_g)$, respectively.

  If $x=v$ then $\tau_2(v_b) = \tau_2(v)$ implies
  $\tau_1(v_b) = \tau_1(v)$. Thus, there is $u_1 \in V$ with
  $\topcorner{1}{u_1}=\leftcorner{1}{x}$,
  i.e. $\left<v,v_b,u_1\right>$ is a counter-clockwise facial
  cycle. Since moving $x$ to $y=\tau_2(x_g)=\tau_1(x_g)$ respects the
  order along $\dside{}{x_g}$, it follows that $\leftcorner{1}{x}$ not
  on $\vside{1}{u}$. Thus $u_1 \neq u$.

  Assume now that there is a counter-clockwise oriented facial cycle
  $\left<w,x,x_g\right>$. Since moving $x$ to
  $y=\tau_2(x) = \tau_2(w) = \tau_1(w) = \ytop(x_g)$ respects the
  order along $\dside{}{x_g}$ it follows that $\rightcorner{1}{w}$ is
  not on $\dside{1}{x_g}$ and $\topcorner{1}{x}$ is strictly above
  $\ytop(x_g)$. Now, if $x=v_b$ then $w=v$ and
  $\tau_1(w) = \tau_2(w) = \tau_2(v_g) = \tau_1(v_g)$. Thus there is a
  clockwise facial cycle $\left<w,v_g,u_2\right>$ with
  $\topcorner{1}{u_2} = \rightcorner{1}{w} \neq \topcorner{1}{x_g}$. It
  follows that $u \neq x_g$

  If $x=v$ then $w=v_g$ and $\tau_2(v_b)=\tau_2(x)$ implies
  $\tau_1(v_b)=\tau_1(x)$.  Thus there is a clockwise facial cycle
  $\left<x,v_b,u_1\right>$ with $u_1 \neq x_g$.

  It remains to show property DTL: Let $(v,w)$ be an edge of $\DAG_r(T)$.
\item[Case $y < \tau_1(x)$:]
     \begin{figure}
    \centering
    \begin{subfigure}{.5\textwidth}\centering
    \includegraphics[page=36]{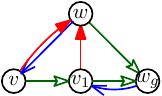}
    \subcaption{$\tau(v) \LEorL \ytop(v_1) = \tau(w)$ by construction}
    \end{subfigure}
    \begin{subfigure}{.45\textwidth}\centering
    \includegraphics[page=37]{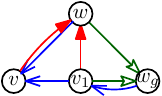}
    \subcaption{$\tau(v) \LEorL \tau(v_1)$ by induction and $\tau(v_1) < \tau(w)$ by induction}
  \end{subfigure}
    \caption{\label{FIG:correctness_vi}Property DTL in the case
    $y < \tau(x)$.}
  \end{figure}
  We use induction on the
  following lexicographical ordering of the edges of $\DAG_r(T)$:
  Consider any ordering $\prec$ extending the partial order given by
  $\DAG_r(T)$. An edge $(v_1,w_1)$ is considered before an edge
  $(v_2,w_2)$ if $w_1 \prec w_2$ or $w_1 = w_2 =: w$ and $v_1$ comes
  before $v_2$ in the clockwise order around $w$ starting from
  $(w_g,w)$.
                   
    Since $\tau_1$ denotes the y-coordinates of the horizontal sides
    in $R_1$, it follows that $\tau_1(v) \leq \tau_1(w)$. Moreover
    $\tau$ never increases for any vertex. Thus, if $\tau(w)$ does not
    change then
    $\tau_2(v) \leq \tau_1(v) \LEorL \tau_1(w) = \tau_2(w)$ for any
    descendant $v$ of $w$. This is especially true if $w$ is the first
    vertex with incoming edges with respect to $\prec$, i.e., the
    first vertex after $X_b$ and~$X_g$.
    
    Let now $(v,w)$ be an edge of $\DAG_r(T)$ such that $\tau(w)$
    changes. If $(v,w)$ is a green edge, i.e. if $v=w_g$ then either
    $w=x$, and $\tau_2(v) = \tau_1(v) \leq y = \tau_2(w)$ (where
    equality is only allowed if $\rightcorner{}{w}$ is the bottommost
    corner on $\dside{}{v}$ and $\leftcorner{}{v}$ and
    $\leftcorner{}{w}$ are not both on the vertical side of the same
    triangle, i.e. if $\left<w,v,v_b\right>$ is an oriented face in $T$), or
    $\tau(w)$ is changed according to the ratio on
    $\dside{}{v}$. Thus, the relationship $\tau(v) \LEorL \tau(w)$ is
    maintained.

    Consider now the case that $(v,w)$ is an incoming red edge of $w$
    or an outgoing (in $T$) blue edge of $w$. Let
    $v = v_0, v_1, \dots, v_k = w_g$ be the neighbors of $w$ in
    counter-clockwise order from $v$ to $v_g$. Observe that $(v_i,w)$,
    $i=1,\dots,k-1$ (if any) are
    red.  By induction, we have
    $\tau_2(v_i) \LEorL \tau_2(w)$, $i=1,\dots,k$. Observe further that
    for any $i=1,\dots,k$ either $(v_{i-1},v_{i})$ is a green edge or
    $(v_{i},v_{i-1})$ is a blue edge. See~\cref{FIG:correctness_vi}.

    Assume first that $k=1$, i.e. $v_1=w_g$. If $w=x$ then
    $\tau_2(v) \leq \tau_1(v) < y = \tau_2(x)$. If $w\neq x$ and
    $(v,w_g)$ is green then both, $\tau(v)$ and $\tau(w)$ have been
    modified according to the ratio on $\dside{}{w_g}$. Hence, the
    relationship $\tau(v) < \tau(w)$ is maintained.

    If $k \geq 1$ and $(v_{1},v)$ is blue we can apply the
    inductive hypothesis on both, $(v,v_1)$ and $(v_1,w)$, and get
    $\tau_2(v) \LEorL \tau_2(v_1) \LEorL \tau_2(w)$.
    Finally, if $k>1$ and $(v_0,v_1)$ is green
    then $\tau(v_0)$ was set according to the ratio on the diagonal of
    $\tau(v_1)$ which implies especially that
    $\tau_2(v_0) \LEorL \tau_2((v_1)_r) = \tau_2(w)$.
    \begin{figure}
      \begin{tabular}{cccc}
        \begin{subfigure}{.22\textwidth}\centering
          \includegraphics[page=39]{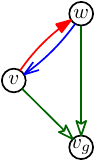}
          \subcaption{\label{}$\tau(v) < \tau(w)$ by construction}
        \end{subfigure}
        &
          \begin{subfigure}{.22\textwidth}\centering
            \includegraphics[page=40]{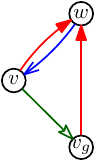}
            \subcaption{\label{}$\tau(v) \LEorL \ytop(v_g) = \tau(w)$ by construction}
          \end{subfigure}
        &
          \begin{subfigure}{.22\textwidth}\centering
            \includegraphics[page=41]{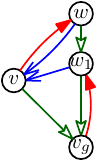}
            \subcaption{\label{}$\tau(w_1) \LEorL \tau(w)$ by construction and \dots}
          \end{subfigure}
        &
          \begin{subfigure}{.22\textwidth}\centering
            \includegraphics[page=42]{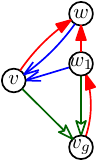}
            \subcaption{\label{}$\tau(w_1) < \tau(w)$ by induction and \dots}
          \end{subfigure}\\
        && \multicolumn{2}{c}{\dots $\tau(v) \LEorL \tau(w_1)$ by induction.}
      \end{tabular}
      \caption{\label{FIG:correctness_wi}Property DTL in the case $y > \tau(x)$.}
    \end{figure}
  \item[Case
    $y > \tau_1(x)$.]  For the case $y > \tau_1(x)$, we use the following
    order for the induction. An edge $(v_1,w_1)$ is considered before
    an edge $(v_2,w_2)$ if $v_2 \prec v_1$ or $v_1 = v_2 =: v$ and
    $w_1$ comes before $w_2$ in the counter-clockwise order around $v$
    starting from $(v,v_g)$.

    Since $y < \tau(X_r)$ and $\tau$ is always updated according to
    the ratio on the incident diagonal, it follows immediately by
    induction on the order in which the vertices are processed that
    $\tau_2(v) < \tau_2(X_r)$ for all vertices $v \neq X_r$. Hence, this
    is especially true if $v$ is the last vertex according to $\prec$
    before $X_r$. 

    Further, observe that $\tau$ is never decreased if $\tau(x)$ is
    increased. Thus, if $\tau(v)$ was not changed then we have
    $\tau_2(v) = \tau_1(v) \LEorL \tau_1(w) \leq \tau_2(w)$.  Assume
    now that $w \neq X_r$ and that $\tau(v)$ has changed.

    If $(v,w)$ is a green edge, i.e. if $v=w_g$ then $\tau(w)$ is
    changed accordingly to the ratio on the diagonal of
    $\Delta(v)$. Thus, the relationship $\tau(v) \LEorL \tau(w)$ is
    maintained.

    Consider now the case that $(v,w)$ is an outgoing red edge of $v$
    or an incoming (in $T$) blue edge of $v$. Let
    $w = w_0, w_1, \dots, w_k=v_g$ be the neighbors of $v$ in
    clockwise order from $w$ to $v_g$~--~where $k=1$ is
    possible. Observe that $(w_i,v)$, $i=1,\dots,k-1$ (if any) are
    blue. See~\cref{FIG:correctness_wi}. Observe further that for any
    $i=1,\dots,k$ either $(w_{i-1},w_{i})$ is a green edge or
    $(w_{i},w_{i-1})$ is a red edge.

    If $k=1$ and $v=x$ then
    $\tau_2(v) = y \LEorL \tau_1(w) \leq \tau_2(w)$, where equality is
    only possible if $(w,v)$ is blue and $(v_g,w)$ is red.

    We distinguish four more cases: ($k=1$ and $v\neq x$) and $k>1$
    and for each of them whether the edge between $w$ and $w_1$ is
    green or red.  We start with $k=1$: If $(w,v_g)$ is red then
    $\tau_2(v) \LEorL \tau_2((v_g)_r) = \tau_2(w)$. If $(w,v_g)$ is green
    then $w$ and $v$ are both neighbors on the diagonal of
    $\Delta(v_g)$ and the relationship $\tau(v) < \tau(w)$ is
    maintained. Consider now $k>1$: Since $w_i$, $i=1,\dots,k-1$ is
    before $w$ in the counter-clockwise ordering around $v$ after
    $v_g$, we have by induction that $\tau_2(v) \LEorL \tau_2(w_i)$. If
    $(w_1,w)$ is red, we observe that $v \prec w_1$ and we have again
    by induction that $\tau_2(w_1) < \tau_2(w)$. If $(w,w_1)$ is green
    then either $\tau(w_1)$ had not been changed at all or $\tau(w)$
    was changed accordingly, thus maintaining the relationship
    $\tau(w_1) \LEorL \tau(w)$. Thus, in both cases, we have
    $\tau_2(v) \LEorL \tau_2(w)$.
  \item[$\dside{2}{x_g}=\dside{1}{x_g}$:] Observe that all descendants
    of $x_g$ in $\DAG_r(T)$ are on the blue path $p_b$ or in the
    exterior of the cycle $C'$. Further, observe that the horizontal
    sides of the vertices in $p_b$ can be constructed without knowing
    anything about the interior of $C'$. Thus $\hside{}{x_g}$ does not
    change.  Since $(x_g)_r$ is also in the exterior of $C'$ or on
    $C'$ it follows that also the height of $\Delta(x_g)$ remains
    unchanged. Thus $\dside{}{x_g}$ is not changed.
  \item[$\left<R_1,R_2\right>$ is a linear morph:]
    \cref{COND:inCHAR_parallel} in \cref{THEO:charRTmorph} is always
    fulfilled if the side is vertical or horizontal. Thus it suffices
    to consider the diagonal sides.
    The only vertex $v$
    for which the ratio $\lambda(v)$ changed was $v=x$. However, since
    $\dside{}{x_g}$ did not change it follows that $\dside{1}{x_g}$
    and $\dside{2}{x_g}$ are parallel. Thus, \cref{THEO:charRTmorph}
    implies that $\left<R_1,R_2\right>$ is a linear morph.  
   \end{description}
\end{proof}

\ordering
\begin{proof}
  We prove the lemma by induction on $n$. There is nothing to show if
  $n=1$. If $n>1$ then we apply the inductive hypothesis to
  $P \setminus \{p_n\}$ and $Q \setminus \{q_n\}$ which yields
  $p_i \neq q_i$, $i < n$ such that no element of
  $P \setminus \{p_n\}$ is between $p_i$ and $q_i$. If $p_n$ is not
  between $p_i$ and $q_i$ we are done. Otherwise we have
  $ \dots p_i < \dots < p_n < \dots < q_i < \dots < q_n$. Thus,
  $p_n \neq q_n$ and no element of $P$ can be between $p_n$ and $q_n$.
\end{proof}

\begin{figure}
  \begin{center}
    \includegraphics[page=43]{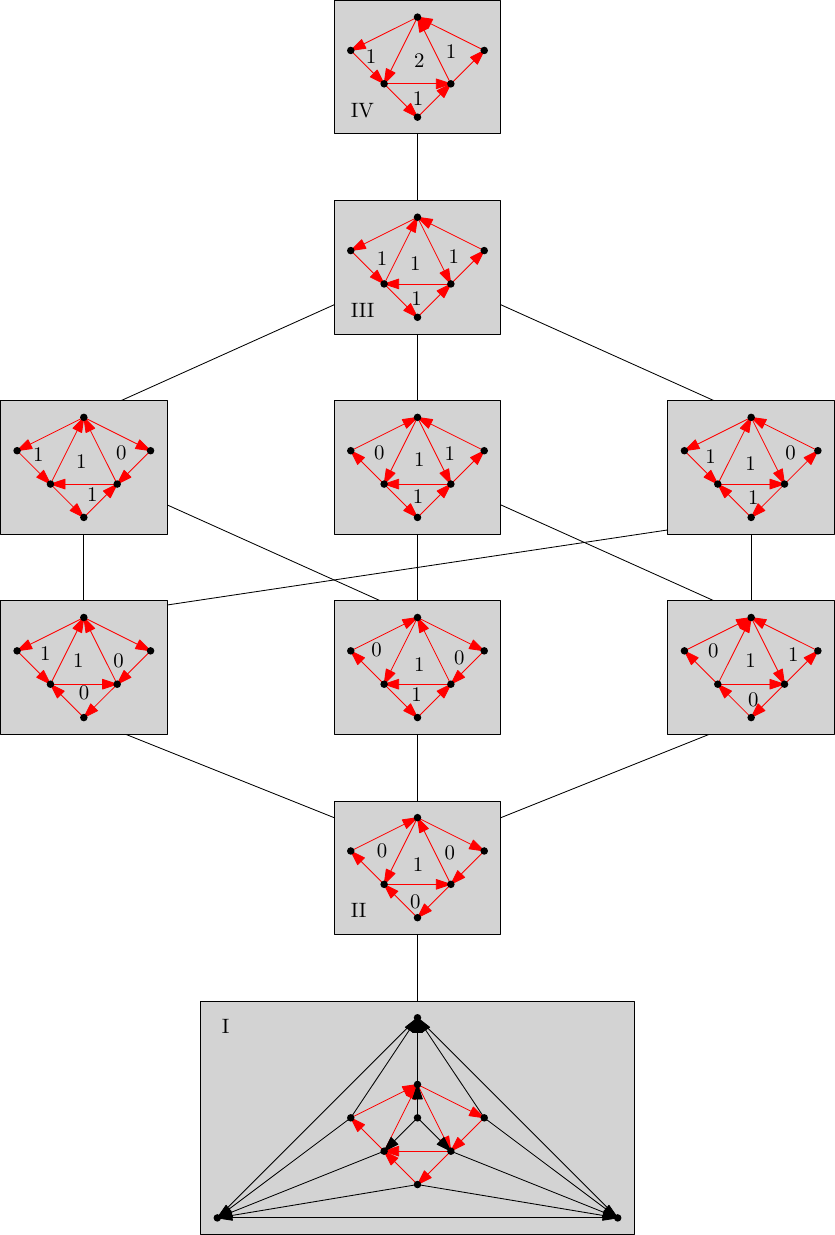}
  \end{center}
  \caption{\label{FIG:lattice} The lattice of all 3-orientations
    of a graph. The whole graph is only drawn in the minimum
    3-orientation. The rigid black edges are not repeated in the
    other drawings~--~they do not change their direction. Face
    labels indicate potentials. Observe that the inner red face is
    actually a separating triangle (there are not-drawn rigid black
    edges inside). }
\end{figure}

\end{document}